\documentclass[lettersize,journal]{IEEEtran}
\usepackage{amsmath,amsfonts, amsthm}
\usepackage{array}
\usepackage{subcaption} 
\usepackage{textcomp}
\usepackage{stfloats}
\usepackage{url}
\usepackage{verbatim}
\usepackage{graphicx}
\usepackage{mathtools}
\usepackage{multirow}
\usepackage{algorithm}
\usepackage[noend]{algpseudocode}
\usepackage{adjustbox}

\newtheorem{theorem}{Theorem}[section]

\newtheorem{lemma}{Lemma}[section]

\newtheorem{remark}{Remark}

\usepackage{cite}
\usepackage{subcaption}

\usepackage[colorlinks=true, linkcolor=blue, citecolor=blue, urlcolor=blue]{hyperref}
\def\BibTeX{{\rm B\kern-.05em{\sc i\kern-.025em b}\kern-.08em
    T\kern-.1667em\lower.7ex\hbox{E}\kern-.125emX}}
\usepackage{balance}

\DeclarePairedDelimiter\abs{\lvert}{\rvert}
\DeclarePairedDelimiter\norm{\lVert}{\rVert}
\newcommand{\bignorm}[1]{\left\lVert#1\right\rVert}

\DeclareMathOperator*{\argmax}{arg\,max}
\DeclarePairedDelimiter\ceil{\lceil}{\rceil}

\algblock{ParFor}{EndParFor}
\algnewcommand\algorithmicparfor{\textbf{parfor}}
\algnewcommand\algorithmicpardo{\textbf{do}}
\algnewcommand\algorithmicendparfor{\textbf{end\ parfor}}
\algrenewtext{ParFor}[1]{\algorithmicparfor\ #1\ \algorithmicpardo}
\algrenewtext{EndParFor}{\algorithmicendparfor}

\def\itemautorefname~#1\null{#1\null}

\long\def\comment#1{}

\newfont{\bbb}{msbm10 scaled 700}

\newfont{\bb}{msbm10 scaled 1000}

\newcommand{\EE}{\mbox{\bb E}}

\newcommand{\Prob}{\mbox{\bb P}}

\newcommand{\bv}{{\bf b}}
\newcommand{\cv}{{\bf c}}

\newcommand{\ev}{{\bf e}}
\newcommand{\fv}{{\bf f}}
\newcommand{\gv}{{\bf g}}

\newcommand{\pv}{{\bf p}}
\newcommand{\qv}{{\bf q}}
\newcommand{\rv}{{\bf r}}

\newcommand{\uv}{{\bf u}}

\newcommand{\xv}{{\bf x}}
\newcommand{\yv}{{\bf y}}

\newcommand{\zerov}{{\bf 0}}

\newcommand{\Am}{{\bf A}}

\newcommand{\Dm}{{\bf D}}
\newcommand{\Em}{{\bf E}}

\newcommand{\Id}{{\bf I}}

\newcommand{\Lm}{{\bf L}}
\newcommand{\Mm}{{\bf M}}

\newcommand{\Pm}{{\bf P}}
\newcommand{\Qm}{{\bf Q}}

\newcommand{\Um}{{\bf U}}

\newcommand{\Xm}{{\bf X}}

\newcommand{\Zm}{{\bf Z}}

\newcommand{\Ac}{{\cal A}}
\newcommand{\Bc}{{\cal B}}

\newcommand{\Ec}{{\cal E}}

\newcommand{\Gc}{{\cal G}}

\newcommand{\Nc}{{\cal N}}
\newcommand{\Oc}{{\cal O}}
\newcommand{\Pc}{{\cal P}}

\newcommand{\Rc}{{\cal R}}
\newcommand{\Sc}{{\cal S}}

\newcommand{\Vc}{{\cal V}}

\newcommand{\deltav}{\hbox{\boldmath$\delta$}}

\newcommand{\Phim}{\hbox{\boldmath$\Phi$}}

\newcommand{\diag}{{\hbox{diag}}}
\renewcommand{\det}{{\hbox{det}}}

\renewcommand{\arg}{{\hbox{arg}}}

\begin{document}
\title{Graph-based Scalable Sampling of 3D Point Cloud Attributes}
\author{%
  \IEEEauthorblockN{%
    Shashank N. Sridhara,
    Eduardo Pavez,
   Ajinkya Jayawant, 
   Antonio Ortega,
  Ryosuke Watanabe, and 
  Keisuke Nonaka%
  }
\thanks{
This work was funded in part by KDDI Research, Inc. and NSF under grant NSF CNS-1956190.

S. N. Sridhara, E. Pavez, A. Jayawant, and A. Ortega are  with the Ming Hsieh Department of Electrical and Computer Engineering, University of Southern California, Los Angeles, 90089, United States
(email: nelamang@usc.edu; pavezcar@usc.edu; jayawant@alumni.usc.edu; aortega@usc.edu).

R. Watanabe and K. Nonaka are with KDDI Research, Inc., 2-1-15 Ohara, Fujimino, Saitama, 356-8502, Japan (email: ru-watanabe@kddi.com; ki-nonaka@kddi.com).

Source code: \url{https://github.com/STAC-USC/3D_point_cloud_sampling}
}

}



\markboth{under review}%
{Graph-based Scalable Sampling of 3D Point Cloud Attributes}

\maketitle
\begin{abstract}
3D Point clouds (PCs) are commonly used to represent 3D scenes. They can have millions of points, making subsequent downstream tasks such as compression and streaming computationally expensive. PC sampling (selecting a subset of points) can be used to reduce complexity. 
Existing PC sampling algorithms focus on preserving geometry features and often do not scale to handle large PCs. In this work, we develop scalable graph-based sampling algorithms for PC color attributes, assuming the full geometry is available. 
Our sampling algorithms are optimized for a signal reconstruction method that minimizes the graph Laplacian quadratic form. We first develop a global sampling algorithm that can be applied to  PCs with millions of points by exploiting sparsity and sampling rate adaptive parameter selection. 
Further, we propose a block-based sampling strategy where each block is sampled independently. 
We show that sampling the corresponding sub-graphs with optimally chosen self-loop weights (node weights) will produce a sampling set that approximates the results of global sampling while reducing complexity by an order of magnitude.
Our empirical results on two large PC datasets show that our algorithms outperform the existing fast PC subsampling techniques (uniform and geometry feature preserving random sampling) by $2$dB. Our algorithm is up to $50$ times faster than existing graph signal sampling algorithms while providing better reconstruction accuracy. Finally, we illustrate the efficacy of PC attribute sampling within a compression scenario, showing that pre-compression sampling of PC attributes can lower the bitrate by $11\%$ while having minimal effect on reconstruction.

\end{abstract}


\section{Introduction}
\label{sec:intro}

\IEEEPARstart{P}{oint} clouds (PCs) have become a popular data format to represent 3D objects and scenes. A PC is comprised of a list of unordered points in $\mathbb{R}^{3}$, $\Pc = \{ \mathbf{p}_{i} = (x_{i}, y_{i}, z_{i}) \}$, representing the geometry, and their corresponding attributes such as color, intensity, or surface normals. A typical PC may contain millions of points. As a result, subsequent downstream tasks, such as compression and real-time streaming in telepresence, become computationally costly \cite{Mekuria2017_immersivecomm, Zhang2021_volumetricSR, Zhang2022_volumetric_streaming }. To alleviate the computational load, we can employ PC sampling (selecting a subset of representative points), compress and stream a low-resolution PC, and reconstruct the original high-resolution PC. A similar technique is employed in image and video compression, i.e., chroma sub-sampling,  where chrominance channels are sampled and encoded at a lower resolution \cite{chen2009_chromaImages}.
 
Existing PC sampling frameworks mainly focus on preserving geometry features (e.g., edges, corners) \cite{chen2017_pcsamplgraph, lang2020_samplenet}, which can be suboptimal when the goal is to accurately represent the attributes, as shown in our experiments (\autoref{sec:experiments}).
This work focuses on scenarios where only attributes are sampled, while geometry information is not sampled (i.e., all point coordinates are preserved). For example, 
in PC compression applications, geometry is available at both the encoder and the decoder and can be used to compress the color attributes. Thus, some color channels (chrominance) can be sampled at the encoder and reconstructed/interpolated to the original higher resolution at the decoder \cite{sridhara2022_chromapc}.
PC attribute sampling is also crucial in immersive communication requiring low latency  \cite{Mekuria2017_immersivecomm} or rendering at varying resolutions \cite{Zhang2021_volumetricSR, Zhang2022_volumetric_streaming}.

\begin{figure}[t]
     \centering
         \includegraphics[width=\linewidth]{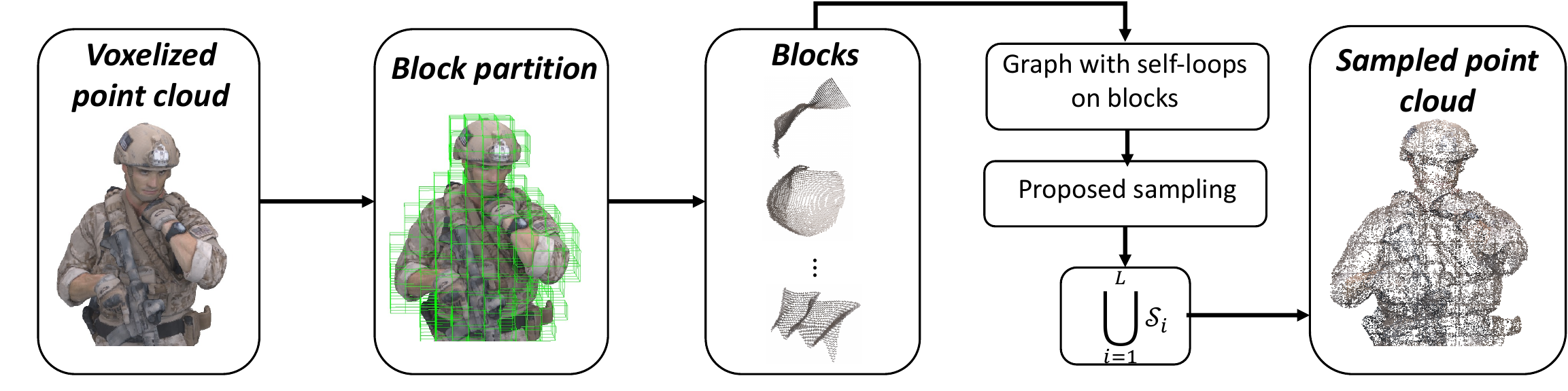}
        \caption{Summary of the proposed reconstruction aware block sampling (RABS) algorithm.  RABS first divides the PC into blocks, then self-loops are added to sub-graphs, followed by sub-graph sampling.}
\label{fig:summary}
\end{figure}


Signal sampling and reconstruction are fundamental operations in signal processing. For 1D signals and images
we can employ regular sampling patterns (e.g., selecting even samples). In our preliminary work, we proposed an  \textit{uniform sampling} algorithm, which selects points in the PC belonging to a regular 3D grid \cite{sridhara2022_chromapc}. 
While this approach is computationally inexpensive, its performance suffers due to several factors: 1) the sampling algorithm is unaware of the reconstruction algorithm, 2)  sampling can only be performed at certain fixed rates, e.g.,  $50\%$, $25\%$, $12.5\%$, etc., 3) regular sampling is sub-optimal for PCs with irregular point distribution and non-uniform point densities   (e.g., in LiDAR PCs \cite{li2021_lidarPC, sridhara2021_lidarpc}). All these factors make 3D PC sampling challenging.
 
In this paper, we formulate the problem of sampling set selection with two goals.
First, the sampling set should be optimized for both the reconstruction algorithm and the sampling rate.
Second, the sampling algorithm should scale to PCs with millions of points, often encountered in practice \cite{chou2017_8idataset,loop2016_mvub}.   

3D PCs can be represented and processed using graphs,  with nodes 
corresponding to points, edge weights chosen based on the relative positions of the points, and PC attributes are treated as graph signals \cite{Shuman2013_gsp, ortega2022_gsp}. 
This has led to the development of graph signal processing (GSP) techniques for PC denoising \cite{watanane2022_pcdenoising}, attribute compression \cite{hong2022_fractional, hong2022_motion} and sampling \cite{dinesh2023_pcsamplBSGDA,tanaka2020_samploverview}.  
However, even existing fast (eigendecomposition-free) graph signal sampling algorithms \cite{sakiyama2019_edfreesampling, jayawant2021_avm, wang2019_samplAopt, puy2018_randomsampl, puy2018_structuredsampl, bai2020_bsgda} can generally only be applied to small/medium graphs (1000-10,000 nodes) \cite{jayawant2021_avm} and are optimized for bandlimited graph signal reconstruction. 

This work proposes a computationally efficient \textit{reconstruction-aware global sampling (RAGS) algorithm}  for PC attributes. 
Our RAGS algorithm is optimized for a popular and scalable graph Laplacian regularized (GLR) reconstruction method 
\cite{ Zhu2002LearningFL, bai2020_bsgda, dinesh2019_pcgeorecon}. 
Similar to \cite{sakiyama2019_edfreesampling, jayawant2021_avm}, we choose samples such that certain interpolating vectors at the sampled locations are as close to orthogonal as possible. 
In \cite{sakiyama2019_edfreesampling, jayawant2021_avm},  interpolators based on spectral filters are used, for which it is necessary to estimate the cut-off frequency. 
Instead, in our work, we construct polynomial interpolators directly in the vertex domain without requiring graph frequency analysis. 
Unlike prior work, our proposed interpolators   \textit{adapt to the sampling rate} by means of their polynomial degree, $p$. 
This allows us to select highly localized (small $p$) or less localized (large $p$) interpolators depending on whether the sampling rate is high (e.g., over 25\%) or low.   We propose a method to choose the localization parameter $p$ as a function of the sampling rate, which allows our RAGS algorithm to handle large PCs, even at higher sampling rates.
Our proposed sampling algorithm has a complexity of $\Oc((\Bar{d}^{3p} + s)N)$, where $\Bar{d}-1$ is the maximum unweighted graph degree,  $s$ is the number of samples and $N$ is the total number of points. 

Since PCs can be partitioned into non-overlapping blocks using the octree structure with $\mathcal{O}(N\log(N))$ complexity \cite{jackins1980_octtrees, pavez2018_polygoncloudcompr}, we propose a  \textit{reconstruction-aware block sampling (RABS) algorithm}, which samples each block independently (see  \autoref{fig:summary}).
Because independent block sampling ignores connections between points across block boundaries, we propose to add node weights (self-loops) to those points. We show that RABS on graphs with self-loops reduces the sampling complexity and the runtime by an order of magnitude compared to the RAGS with minimal loss in reconstruction quality. 

Our main contributions are:
\begin{itemize}
    \item A global sampling set selection algorithm optimized for a graph Laplacian-based reconstruction(\autoref{sec:sampling_algo_development}).
    \item Methods to select the interpolation localization parameter $p$ as a function of the sampling rate (\autoref{subsec:effect_of_p}). 
    \item A block-based graph sampling framework where self-loops are added to block-boundary nodes (\autoref{sec:block_pcsampl}). 

    \item A comprehensive evaluation of our algorithms on two large PC datasets and their application to \textit{PC attribute compression through chroma subsampling}  (\autoref{sec:experiments}).

\end{itemize}

\subsection{Related work}
\subsubsection{Graph signal sampling}
Most existing graph signal sampling algorithms have been developed for bandlimited graph signals \cite{tanaka2020_samploverview}. Early methods required computing the full eigendecomposition (ED) \cite{chen2015_graphsampling} or computing one eigenvector in each step \cite{anis2016_spectralproxies}, making them only suitable for small graphs. 
Random sampling \cite{puy2018_randomsampl, puy2018_structuredsampl} is the fastest ED-free algorithm, but it can result in subpar reconstruction quality \cite{tanaka2020_samploverview}.
Deterministic ED-free sampling algorithms \cite{sakiyama2019_edfreesampling, jayawant2021_avm, wang2019_samplAopt} provide better reconstruction quality. However, all ED-free sampling algorithms developed for bandlimited graph signals, including random sampling \cite{puy2018_randomsampl}, have to estimate the bandwidth (cut-off frequency) to design spectral filter parameters, which makes them too complex for large graphs (e.g., more than $10,000$ nodes). 
While our proposed block-based sampling approach can be combined with existing ED-free sampling algorithms \cite{sakiyama2019_edfreesampling, jayawant2021_avm, wang2019_samplAopt}, their complexity would still be high owing to the use of spectral filtering and the need to estimate certain parameters such as 
cut-off frequencies (\autoref{subsec:comparison}), or coverage sets \cite{bai2020_bsgda}, while often also having poor reconstruction accuracy at higher sampling rates (see \autoref{subsec:comparison}).
Because our proposed algorithms, RAGS and RABS,  are based on a graph Laplacian regularizer (GLR) reconstruction \cite{chapelle2006_labelprop}, we do not require spectral filtering or bandwidth estimation.  

\subsubsection{Point cloud sampling}
The existing literature revolves around methods that preserve geometry features \cite{chen2017_pcsamplgraph, lang2020_samplenet,dinesh2023_pcsamplBSGDA}. 
 \cite{chen2017_pcsamplgraph, lang2020_samplenet} use the sampled PC for registration and classification without requiring attribute reconstruction. More recently, \cite{dinesh2023_pcsamplBSGDA,bai2020_bsgda}  proposed a graph-based sampling set selection algorithm that only considers geometry reconstruction and has relatively high computational complexity.

The rest of the paper is organized as follows. GSP preliminaries and notations are given in \autoref{sec:preliminaries}. We develop the sampling algorithm and discuss its properties in \autoref{sec:sampling_algo_development}. We propose block-based sampling and provide a detailed complexity analysis in \autoref{sec:block_pcsampl}. Experimental results and conclusions are in  \autoref{sec:experiments} and \autoref{sec:conclusion}, respectively.

\section{Preliminaries}
\label{sec:preliminaries}

We represent sets using calligraphic uppercase: the sampling set $\Sc$ is a subset of $\Vc$ and its complement is $\Rc = \Vc \setminus \Sc$. 
We represent matrices by bold uppercase, $\Xm$, vectors by bold lowercase, $\xv$, and scalars by plain lowercase, $x$. 
The submatrix of $\Xm$ obtained by extracting columns and rows indexed by the sets $\Ac$ and $\Bc$ respectively is denoted by $\Xm_{\Bc \Ac}$.

We denote 3D points representing the geometry of a PC as $\Pc = \{ \mathbf{p}_{i} = (x_{i}, y_{i}, z_{i}) , \forall i \}  \in \mathbb{R}^{N \times 3}$, where $N$ is the number of points. Without loss of generality, we assume that the point attributes are scalar and denoted by $\mathcal{C} = \{ f_{i} \in  \mathbb{R}, \forall i\}$.

A graph  $\Gc = (\Vc, \Ec)$, consists of a vertex set $\Vc$ of size $N$  and an edge set $\mathcal{E}$ connecting vertices. A graph signal is a vector $\fv \in \mathbb{R}^{N}$ whose $i$th entry corresponds to the $i$th node in the graph. We define $w_{ij}>0$ to be the weight of the edge between vertices $i$ and $j$. The adjacency matrix, $\Am \in \mathbb{R}^{N \times N}$, is a symmetric matrix whose $(i,j)$th element is $w_{ij}$, if $i$ and $j$ are connected and $0$ otherwise. The degree matrix $\Dm$ is a diagonal matrix with diagonal entries, $\Dm_{i,i} = \sum_{j} \Am_{i,j}$ \cite{ortega2022_gsp}. 

The combinatorial Laplacian, $\Lm:= \Dm - \Am$ is symmetric and positive semi-definite, with an orthogonal set of eigenvectors, $\mathbf{U} = [\uv_1, \uv_2, \ldots, \uv_N]$ and eigenvalues $0 = \lambda_{1} \leq \lambda_{2} \leq,\cdots,\leq \lambda_{N}$ representing the graph frequencies. 
The random walk Laplacian is defined as $\Lm_{rw}:= \Dm^{-1}\Lm$ \cite{ortega2022_gsp}. 
The generalized graph Laplacian of a graph $\mathcal{G}$ with self-loops is $\Tilde{\Lm} := \Dm - \Am + \Phi$, where $\Phim$ is diagonal and $\mathbf{\Phi}_{ii} \geq 0$ is the self-loop weight at node $i$, which can be interpreted as a  \textit{boundary condition} for that node (see \autoref{fig:1d_dct_bc} and \cite{biyikoglu2007laplacian}).
\begin{figure}[t]
     \centering
         \includegraphics[width=\linewidth]{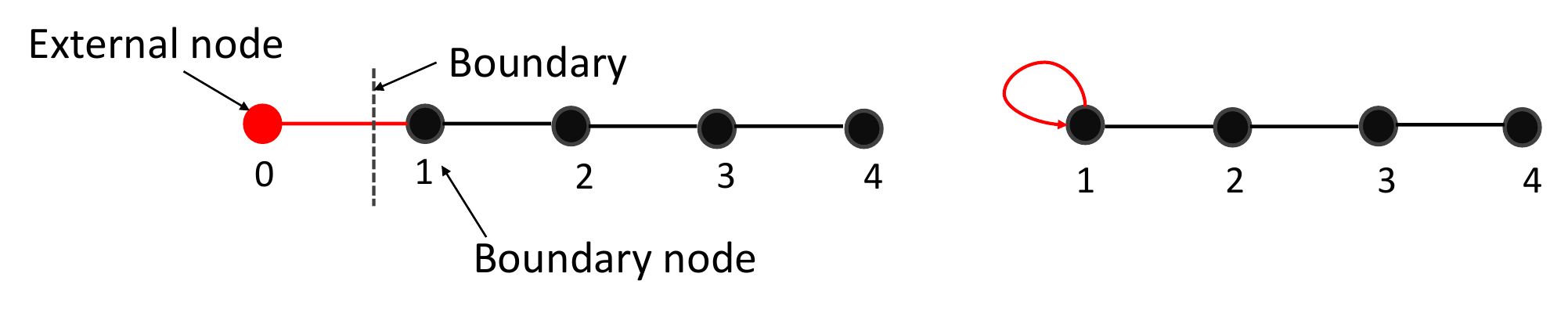}
         \caption{Illustration of self-loop for the boundary node as a consequence of boundary condition using a line graph}
\label{fig:1d_dct_bc}
\end{figure}

We use a $K$-nearest neighbor (KNN) graph to represent the geometric relationships between points because KNN can be implemented efficiently for large PCs. 
However, our methods can be applied with other graph constructions \cite{shekkizhar2020_nnk, watanabe2024fast}. Each point $\mathbf{p}_i$ is associated with a node in the graph $i \in \mathcal{V}$. Two nodes are connected to their  $K$-nearest neighbors based on the  Euclidean distance between points. We symmetrize the KNN graph by connecting two nodes $i$ and $j$ with an edge $(i,j)$ if $i$ is a KNN of $j$, or vice versa. 
The neighborhood of a node $i \in \mathcal{V}$ is denoted by $\mathcal{N}(i)$. The edge weights are computed using the Gaussian kernel with parameter $\sigma>0$, 
  \begin{equation}
    \label{eqn:graph_weights}
     w_{ij} = e^{-\frac{\Vert \mathbf{p}_{i} - \mathbf{p}_{j}\Vert_{2}^2}{2 \sigma^2}}.
 \end{equation}

\section{Sampling algorithm development}
\label{sec:sampling_algo_development}

Given a set of points $\Pc = \lbrace \pv_1,\cdots, \pv_N \rbrace$, and its corresponding attribute (signal) $\fv \in \mathbb{R}^{N}$, we aim to obtain a sampling set, $\Sc \subset \lbrace 1, 2, \cdots, N \rbrace$,  such that we can accurately reconstruct $\fv_{\Rc}$ from the sampled signal $\fv_{\Sc}$, where $\Rc = \Vc\setminus\Sc$.

\subsection{Overview of proposed  PC attribute sampling framework}
\label{subsec:block_based_formulation}

Throughout this paper, we will perform signal reconstruction by minimizing the graph Laplacian quadratic form with sample consistency constraints    \cite{chapelle2006_labelprop, Zhu2002LearningFL, bai2020_bsgda, dinesh2023_pcsamplBSGDA}. Given a sampling set $\Sc$, the reconstructed signal $\hat{\fv}$ is
\begin{equation}
    \label{eqn:recon_minimization}
\hat{\fv} =   \arg\min_{\xv} \xv^{\top} \Lm \xv \quad
    \textrm{s.t.} \quad  \xv_{\Sc} = \fv_{\Sc}.
\end{equation}
$\hat{\fv}$ is the smoothest graph signal consistent with the samples $\fv_{\Sc}$. 
There exist various efficient algorithms to solve \eqref{eqn:recon_minimization}  (e.g.,  \cite{Zhu2002LearningFL} and \cite[Alg. 11.1]{chapelle2006_labelprop}) based on sparse matrix-vector products.  
%

For a desired  sampling rate $\alpha$ ($s/N \leq \alpha$), our goal is to obtain a sampling set $\Sc^{*}$ of size $\vert \Sc^{*} \vert = s$ that gives the best reconstruction according to \eqref{eqn:recon_minimization}.
While reconstruction algorithms to solve \eqref{eqn:recon_minimization} can scale to large PCs, traditional graph-based sampling set selection algorithms are often too complex to be used in the PCs considered in this work. 

Two main ideas allow our approach to scale to large graphs. 
First, similar to  \cite{sakiyama2019_edfreesampling, jayawant2021_avm}, we compute interpolating vectors for each node,   but we avoid bandwidth estimation by using spatially-localized interpolators designed for reconstruction with \eqref{eqn:recon_minimization}. This approach is our Reconstruction Aware Global Sampling (\hyperref[algo:fast_sampling]{RAGS}), which is developed later in this section. Second, we propose using a block-based sampling strategy whereby the PC is divided into non-overlapping blocks as shown in \autoref{fig:summary}, and each block is sampled independently.  
The proposed Reconstruction Aware Block Sampling (\hyperref[algo:block_sampling]{RABS})  can be performed using information local to each block, leading to significant complexity reductions in the complexity of computing the interpolating vectors and selecting the sampling set. 
 In RABS,  
the point set $\Pc$ is partitioned into $L$ non-overlapping blocks (subsets of points) denoted by $\mathcal{B}_{1}, \mathcal{B}_{2}, \dots \mathcal{B}_{L}$ such that $\Pc = \bigcup_{j=1}^{L} \mathcal{B}_{j}$, with $\abs{\mathcal{B}_{j}} = b_{j},  \forall j$.  Let $\{\mathcal{G}_{j}\}_{j=1}^{L}$ be the subgraphs of $\mathcal{G}$ associated with the points in blocks $\{\mathcal{B}_{j}\}_{j=1}^{L}$. 
Our algorithm obtains sampling sets $\{\Sc_{j}\}_{j=1}^{L}$ for each of the subgraphs $\{\mathcal{G}_{j}\}_{j=1}^{L}$, which are aggregated to form a global sampling set $\bigcup_{j=1}^{L} \Sc_{j}$. 
We sample $s_i$ points from the $i$th block so that $\frac{s_i}{b_i} \leq \alpha$, and the global sampling rate is $\alpha$. 

In the rest of this section, we derive the RAGS algorithm, while the RABS algorithm will be introduced in  \autoref{sec:block_pcsampl}. 

\begin{algorithm}[t]
\caption{Reconstruction-aware global sampling (RAGS)}\label{algo:fast_sampling}
\begin{algorithmic}[1]
\Procedure{Sampling}{$\Am, \Dm, \alpha$}
\State $\Zm \gets \frac{1}{2}(\Id + \Dm^{-1} \Am)$
\Comment{one-hop operator}
\State Compute $p$ \Comment{$p$ from $\alpha$}
\State Compute $\Qm(p)$ 

\State Compute $\langle \qv_{i}^{(p)}, \qv_{j}^{(p)} \rangle$ for all $i,j$

\While{$\abs{\Sc}/N \leq  \alpha$} 
\State $x^{*}\gets \argmax_{x \in \Sc^{c}}  \norm{\qv_{x}^{(p)}}_{2}^{2} - \sum_{y \in \Sc} \frac{\langle \qv_{x}^{(p)}, \qv_{y}^{(p)} \rangle^{2}}{\norm{\qv_{y}^{(p)}}_{2}^{2}}$
\State $\Sc \gets \Sc \cup x^{*}$

\EndWhile
\State \textbf{return} $\Sc$
\EndProcedure
\end{algorithmic}
\end{algorithm}

\begin{algorithm}[t]
\caption{Reconstruction-aware block sampling (RABS)}\label{algo:block_sampling}
\begin{algorithmic}[1]
\Procedure{PC-sampling}{$\Pc, \alpha$}
\State $\Am \gets \text{construct\_graph}(\Pc)$
\State $\{\mathcal{B}_{1} \dots \mathcal{B}_{L} \} \gets \text{divide\_PC}(\Pc)$ \Comment{octree subdivision}
\State $\Sc \gets \emptyset$
\For{$i \gets 1, L$} \Comment{sampling blocks}
    \State Obtain $\Am^{(i)}$ \Comment{block adjacency matrix}

    \State $\mathbf{\Phi}^{(i)} \gets \text{selfloop-weights}(\Am, \Bc_{i})$
    \State $\Dm^{(i)} \gets diag(\Am^{(i)}.\mathbf{1} )$

    \State $\Tilde{\Dm}^{(i)} \gets \Dm^{(i)} + \mathbf{\Phi}^{(i)}$ \Comment{add self-loops}

    \Comment{Obtain local sampling set using \hyperref[algo:fast_sampling]{RAGS}}
    \State $\Sc_{i} \gets \text{SAMPLING}(\Am^{(i)}, \Tilde{\Dm}^{(i)}, \alpha)$  
\EndFor
 \State $\Sc \gets \bigcup_{i=1}^{L} \Sc_{i}$
\State \textbf{return} $\Sc$
\EndProcedure
\end{algorithmic}
\end{algorithm}


\subsection{From signal reconstruction to sampling}
\label{subsec:constructing_q}

To design a sampling set selection algorithm, first, from \cite{chapelle2006_labelprop, Zhu2002LearningFL}, we note that the closed-form solution to \eqref{eqn:recon_minimization} is: 
\begin{equation}
\label{eqn:recon_soln}
    \hat{\fv} =  
    \begin{bmatrix}      
    \fv_{\Sc}\\
        - \Lm_{\Rc\Rc}^{-1} 
        \Lm_{\Rc\Sc} \fv_{\Sc}
    \end{bmatrix} 
    = \Mm(\Sc)\fv_{\Sc}. 
\end{equation}
%
Since the interpolation matrix $\Mm(\Sc) \in \mathbb{R}^{N \times \abs{\Sc}}$  is a function of $\Sc$, we are not able to use efficient ED-free sampling set selection methods such as  \cite{sakiyama2019_edfreesampling, jayawant2021_avm}, which define a fixed (sampling-set independent) $N\times N$ matrix of interpolators and choose  $s = |\Sc|$  of its columns to 
optimize a reconstruction error criterion. 
A key contribution of our work is introducing $\Qm(p)$, a parametric approximation of $\Mm(\Sc)$, where the parameter $p$ is a function of $|\Sc|$ rather than $\Sc$. This allows us to use efficient methods to select $|\Sc|$ columns from $\Qm(p)$, to obtain a sampling set $\Sc$ that performs well when interpolating with $\Mm(\Sc)$. 
%
%
%
We start by defining the low-pass graph filter
\begin{equation}
    \label{eqn:1_hop_operator}
    \Zm := \frac{1}{2}(\Id + \Dm^{-1} \Am) = \Id - \frac{1}{2} \Lm_{rw},
\end{equation}
where $\Lm_{rw} = \Dm^{-1} \Lm$ is the random-walk Laplacian.
We define the diagonal node selection matrix $\Pm_{\Ac}$ for any $\Ac \subset \Vc$ as,
\begin{equation}\label{eq_sampling_upsampling}
    [\Pm_{\Ac}]_{ii} =  \begin{dcases}
    1, & \text{if }  i \in \Ac\\
    0,              & \text{otherwise}.
\end{dcases}
\end{equation}
We will require the following lemma (proved in Appendix \ref{app:reconstruction_proof}).
\begin{lemma}
\label{lemma_equivalent_interpolators}
    The closed-form solution in \eqref{eqn:recon_soln} can be equivalently written in terms of $\Zm$ as,
    \begin{equation}
        \label{eqn:recon_soln_graph_filter}
        \hat{\fv}=  (\Id - \Pm_{\Rc} \Zm)^{-1} \Pm_{\Sc} \fv = [\sum_{l=0}^{\infty} (\Pm_{\Rc} \Zm)^{l}]  \Pm_{\Sc} \fv.
    \end{equation}
\end{lemma}
%
%
\begin{figure}[t]
\centering
\begin{subfigure}[b]{0.28\textwidth}
         \centering
         \includegraphics[width=\textwidth]{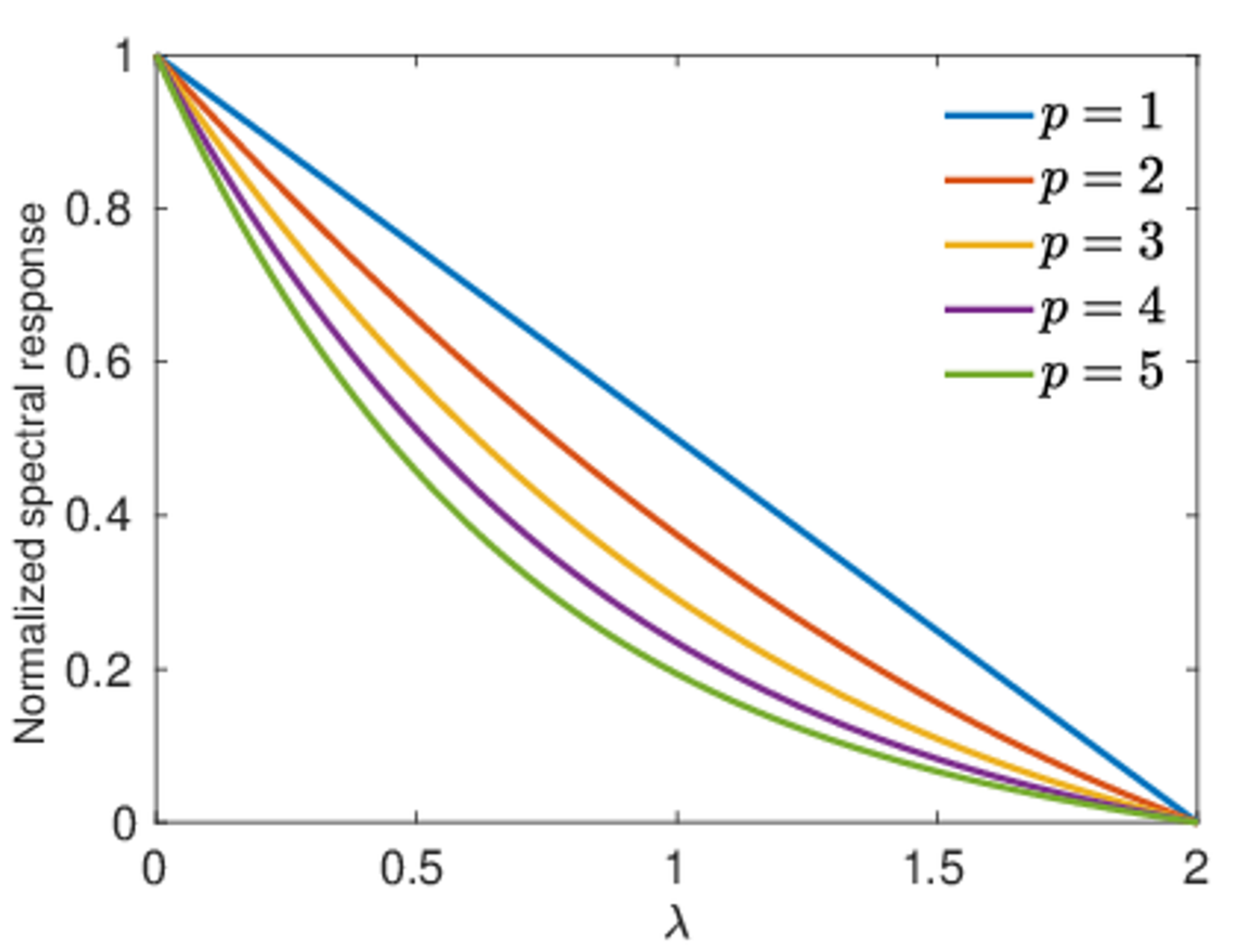}
         \caption{}
         \label{subfig:spectral_response_interpolator}
     \end{subfigure}
     \begin{subfigure}[b]{0.18\textwidth}
         \centering
         \includegraphics[width=\textwidth]{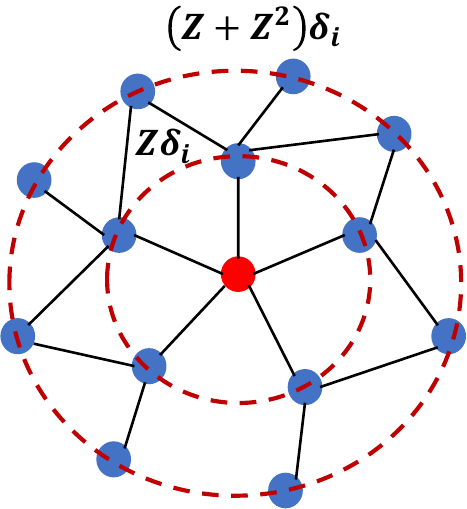}
         \caption{}
        \label{subfig:vertex_intrep_Z}
     \end{subfigure}
     \caption{(a) Normalized spectral response of the proposed interpolator - $g_{p}(\lambda_{i}) = \frac{1}{p} \sum_{l=1}^{p} (1 - \frac{1}{2} \lambda_{i})^{l}$.(b) Vertex domain interpretation of the interpolator $\mathbf{Q}(p)$.} 
\label{fig:interpolator}
\end{figure}
The reconstructed signal $\hat{\fv}$ is a weighted linear combination of the columns of $[\sum_{l=0}^{\infty} (\Pm_{\Rc} \Zm)^{l}]$ and given that    
$\Pm_{\Sc} \fv = \begin{bmatrix}
\fv_{\Sc}^{\top} & \zerov^{\top} \end{bmatrix}^{\top}$, \eqref{eqn:recon_soln_graph_filter} can be approximated by using only the first $p$ terms of the Neumann expansion: 
\begin{equation}
    \label{eqn:approx_fhat}
    \hat{\fv}^{(p)} = \left(\sum_{l=0}^{p} (\Pm_{\Rc} \Zm)^{l} \right)  \begin{bmatrix}
        \fv_{\Sc} \\ \zerov
    \end{bmatrix}.
\end{equation}
Because  \eqref{eqn:approx_fhat} satisfies $\hat{\fv}^{(p)} = \Pm_{\Rc} \Zm \hat{\fv}^{(p-1)} + \Pm_{\Sc} \fv$, we have that $(\hat{\fv}^{(p)})_{\Sc} = \fv_{\Sc}$ thus we can focus on the reconstruction on $\Rc$:
\begin{equation}
\label{eqn:p_interpolation_R}
   \begin{bmatrix}
        \zerov \\ \hat{\fv}^{(p)}_{\Rc}
    \end{bmatrix} =  \Pm_{\Rc} \hat{\fv}^{(p)} = [\sum_{l=1}^{p} (\Pm_{\Rc} \Zm)^{l}]  \begin{bmatrix}
        \fv_{\Sc} \\ \zerov
    \end{bmatrix}.
\end{equation}
Based on this equation, we propose using the matrix
%
\begin{equation}
\label{eqn:interpolator_def}
    \Qm(p) :=  \sum_{l = 1}^{p} \Zm^{l}
\end{equation}
for sampling set selection. Specifically, we claim (see \autoref{subsec:obj_algo}) that for a value of $p$ consistent with the sampling rate (i.e., $p$ large if $|\Sc|$ is small), we have:
\begin{equation}
    \label{eqn:Q_approx}
    \Qm(p)_{\Vc \Sc} \approx \Mm(\Sc).
\end{equation}
%
%
Since $\Qm(p)$ does not depend on $\Sc$, it can be used for sampling.  
%
%
In the spectral domain, filtering with $\Zm$ in \eqref{eqn:1_hop_operator} is given by:
\begin{equation}
    \yv = \Zm \xv =  \Um_{rw} (\Id - \frac{1}{2}\mathbf{\Lambda}) \Um_{rw}^{-1} \mathbf{x},
\end{equation}
thus $\Zm$ is a graph filter  with frequency response $ g(\lambda) = 1 - \frac{1}{2} \lambda$. 
We can also observe that since $\lambda \in [0, 2]$, $g(\lambda) \in [0,1]$.
Similarly, $\Qm(p) = \sum_{l=1}^{p}\Zm^{l}$ is a low-pass filter with frequency response $g_p(\lambda_{i}) = \sum_{l=1}^{p}(1 - \frac{1}{2} \lambda_{i})^l$. Thus, increasing $p$ leads to a sharper frequency response,  as depicted in \autoref{subfig:spectral_response_interpolator}.

The $i$th column of $\Qm(p)$ can be seen as a low-pass filtered delta function $\deltav_{i}$, localized at node $i$. For any $p$ we have:
\begin{equation}
 \label{eqn:p_order_approx}
\qv_{i}^{(p)} := \Qm(p) \deltav_{i} = (\Zm\deltav_{i} + \Zm^{2}\deltav_{i} + \Zm^{3}\deltav_{i} + \cdots + \Zm^{p}\boldsymbol{\delta_{i}}).
\end{equation}
The parameter $p$ defines the localization of the interpolating vectors $\qv_{i}^{(p)}$, i.e., 
when $p = 1$,  $\qv_{i}^{(1)} = \Zm  \boldsymbol{\delta}_{i}$ which results in non-zero values only in the 1-hop neighbors of $i$. Similarly $\qv_{i}^{(2)} = \Zm \deltav_{i} + \Zm^{2} \deltav_{i}$ contains non-zero values within 2-hops of $i$, as shown in \autoref{subfig:vertex_intrep_Z}. In general,    $\qv_{i}^{(p)}$ has non-zero values at nodes within the $p$-hop neighborhood of $i$. 


\subsection{Objective function and sampling algorithm}
\label{subsec:obj_algo}
Starting from the approximation in \eqref{eqn:Q_approx}, we will obtain a sampling set by selecting the $s$ most ``informative''  interpolating vectors, i.e., a subset of column vectors of $\Qm(p)$.
%
From \cite{sakiyama2019_edfreesampling, jayawant2021_avm}, 
the sampling set should be selected such that the interpolators are close to orthogonal. This can be achieved if the overlap between the $p$-hop neighborhoods of two selected interpolating vectors $\qv_{i}^{(p)}$, $\qv_{j}^{(p)}$  $i, j \in \Sc, i \neq j$  is  minimized.
This problem can be approximated by maximizing the volume of parallelepiped formed by the vectors $\{\qv_{i}^{(p)}\}_{i \in \Sc}$, i.e., $\Qm(p)_{\Vc \Sc} $ \cite[Section 4]{jayawant2021_avm}, which is equivalent to \cite{jayawant2021_avm, peng2007determinant}   
\begin{equation}
    \label{eqn:objective_function}
    \Sc^{*} = \argmax_{\Sc : \vert \Sc \vert \leq s} \text{det} (\Qm(p)_{\Vc \Sc}^{\top} \Qm(p)_{\Vc \Sc}).
\end{equation}
We adopt the greedy algorithm proposed in \cite[Appendix D]{jayawant2021_avm} to approximately solve \eqref{eqn:objective_function}, whereby at each iteration, we add to the sampling set the point $x^{*}$ that solves:
\begin{equation}
    \label{eqn:greedy_update}
    x^{*} = \argmax_{x \in \Rc}  \norm{\qv_{x}^{(p)}}_{2}^{2} - \sum_{y \in \Sc} \frac{\langle \qv_{y}^{(p)}, \qv_{x}^{(p)} \rangle^{2}}{\norm{\qv_{y}^{(p)}}_{2}^{2}}.  
\end{equation}
By minimizing \eqref{eqn:greedy_update}, we ensure the interpolator norm $\norm{\qv_{x}^{(p)}}_{2}^{2}$ is large, while the overlap with the already selected points, i.e., $\sum_{y \in \Sc} \langle \qv_{y}^{(p)}, \qv_{x}^{(p)} \rangle^{2}/\norm{\qv_{y}}^{2}, \forall x \in \Rc, y \in \Sc$ is small. 
The proposed reconstruction aware global sampling (RAGS) is presented as \autoref{algo:fast_sampling}, and complexity is discussed in \autoref{subsec:complexity_analysis}.

We now revisit the approximation proposed in \eqref{eqn:Q_approx} given that the sampling set is chosen according to \eqref{eqn:objective_function}.
The overlap between the selected interpolators is dominated by the support of the $p$th order terms of  $\Qm(p)_{\Vc \Sc}$. If we compare it with the $p$th order term in  $\Mm(\Sc)$ we get
\begin{equation}
\label{eqn:interp_error}
    \Em(p) = \Pm_{\Rc} \Zm^{p} \Pm_{\Sc} -  (\Pm_{\Rc} \Zm)^{p} \Pm_{\Sc}.
\end{equation}
The error $\Em(p)$ is zero when $p = 1$. In general, when the samples are far apart, i.e., no two samples in $\Sc$ have common nodes in their $p$-hop neighborhood, the terms $\Pm_{\Rc} \Zm^p \Pm_{\Sc}$ and $(\Pm_{\Rc} \Zm)^{p} \Pm_{\Sc}$ will have the same sparsity pattern leading to small approximation error in \eqref{eqn:interp_error}. 
Although a smaller $p$, in general, reduces the error in \eqref{eqn:interp_error}, it limits the neighborhood information for sampling. For example, for low sampling rates (small $\alpha$) and small $p$, many pairs of interpolators are nearly orthogonal so that two nodes that are relatively close can both be selected for sampling. This can lead to highly non-uniform sampling. Thus, for small $\alpha$, we need a higher value of $p$ for better signal reconstruction. 
%
%

\subsection{Choosing the parameter \texorpdfstring{$p$}{TEXT} based on sampling rate \texorpdfstring{$\alpha$}{TEXT}}
\label{subsec:effect_of_p} 
\begin{figure}[t]
\centering
\begin{subfigure}[b]{0.20\textwidth}
         \centering
         \includegraphics[width=\textwidth]{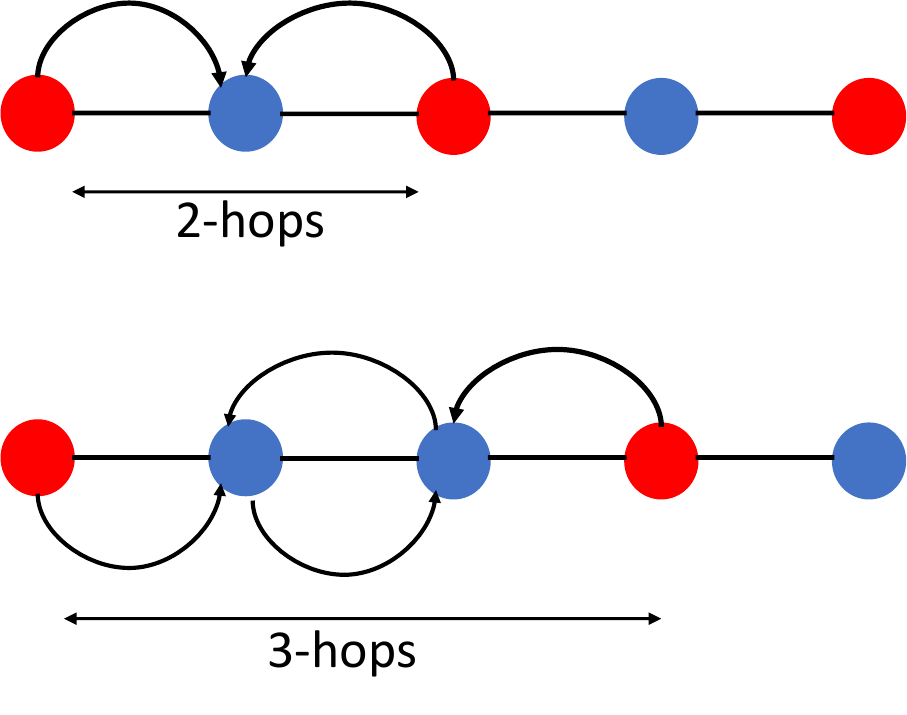}
         \caption{}
         \label{subfig:line_graph}
     \end{subfigure}
     \begin{subfigure}[b]{0.24\textwidth}
         \centering
         \includegraphics[width=\textwidth]{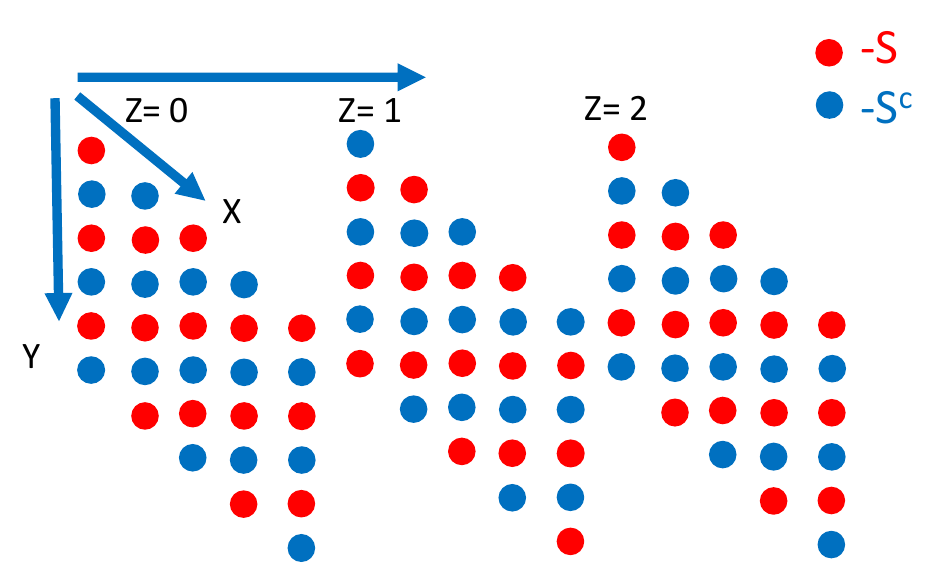}
         \caption{}
        \label{subfig:3d_grid}
     \end{subfigure}
     \caption{Illustration of interpolation when samples are selected uniformly on a line graph at rates $50\%$, $33\%$ (left) and uniform sampling on 3D grid graph at $50\%$ sampling rate (right)} 
\label{fig:choosing_p}
\end{figure}
 From previous discussions, and as will be confirmed by our experiments (\autoref{sec:experiments}), using a larger $p$ leads to sampling sets with better signal reconstruction.   
 However, complexity and runtime also increase with $p$ (see \autoref{subsec:complexity_analysis}). In this section, we investigate the problem of selecting $p$ to achieve a good trade-off between complexity and reconstruction accuracy. In particular, we propose to choose $p$ to be inversely proportional to the sampling rate $\alpha$.
 
To understand the interplay between $p$ and $\alpha$, consider the approximate reconstruction using interpolators in \eqref{eqn:approx_fhat}. Assuming $\Sc$ is known, the parameter $p$ in \eqref{eqn:approx_fhat} represents the minimum number of hops required to interpolate points in $\Rc$. When $\alpha$ is small, points in  $\Rc$ will be far away from sampled points, and thus a higher-order interpolator will be required. In general, the optimal   $p$  depends on  $\Sc$ and $\Zm$, and thus, it cannot be computed before sampling. In what follows, we propose two strategies to choose  $p$  for a given $\alpha$: 1) when $\Sc$ is obtained by uniform sampling on a regular grid graph, and 2) when $\Sc$ is obtained by random sampling on any arbitrary graph.
\subsubsection{Heuristic for uniform sampling of regular graphs}
Consider a uniform sampling set on a line graph (\autoref{subfig:line_graph}). 
When $\alpha=0.5$, at least 1 hop is required to interpolate the points in $\Rc$ from those in $\Sc$. But when $\alpha=0.33$, a minimum of 2 hops is required. This is true in the case of 3D grid graphs where samples are uniformly selected in each dimension as shown in \autoref{subfig:3d_grid}. 
Therefore, we propose selecting  $p$ as the smallest number of hops required for interpolation, that is,
\begin{equation}
    \label{eqn:p_estimate_uni}
    p_{\text{uni}} = \ceil*{\frac{1}{2\alpha}},
\end{equation}
where $\ceil{x}$ is the smallest integer satisfying $x\leq \ceil{x}$.
This choice of $p$  is a good heuristic for dense PCs representing a smooth surface, such as those used in AR/VR applications. 
\begin{figure}[t]
     \centering
         \includegraphics[width=0.35\textwidth]{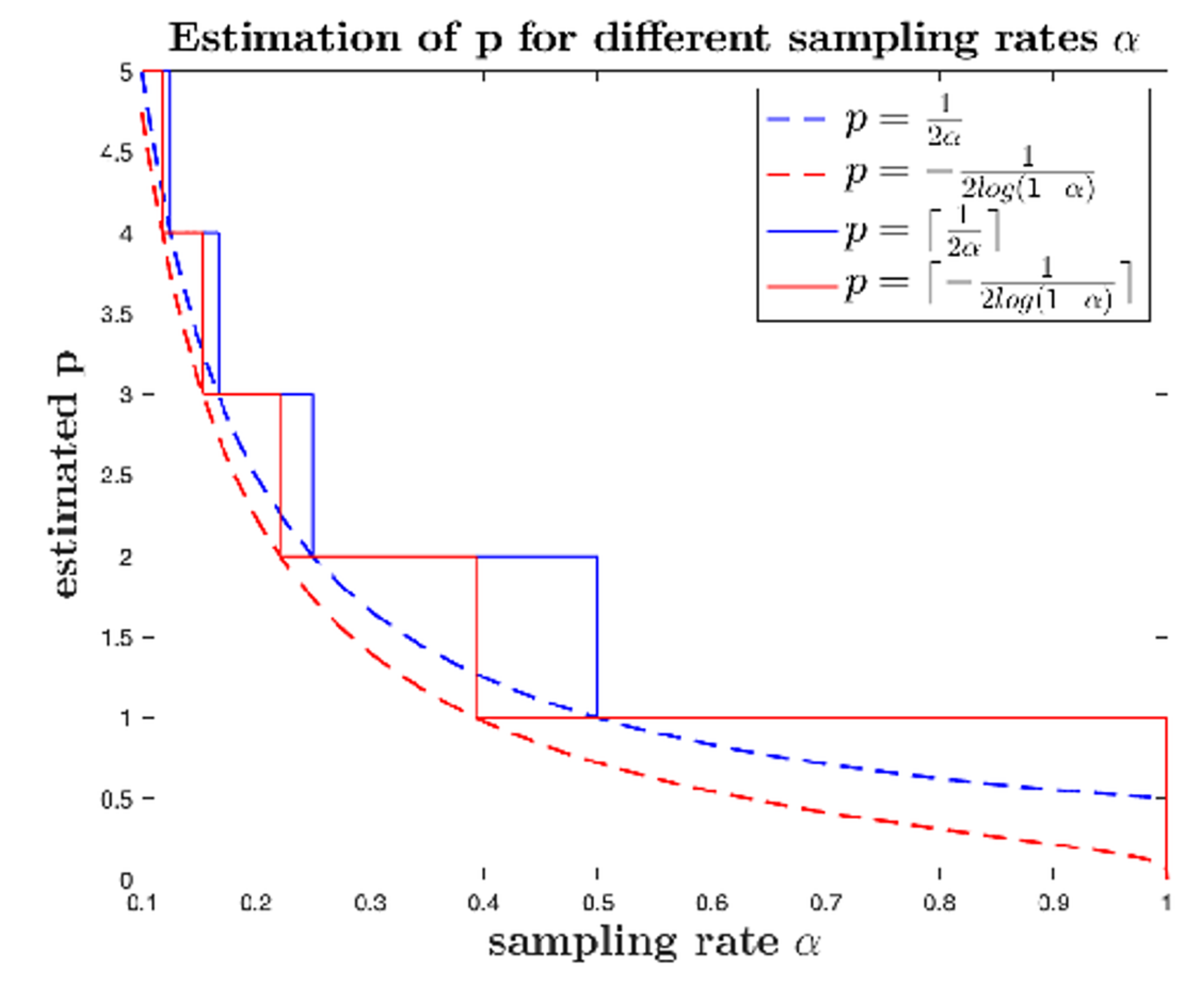}
        \caption{Proposed $p$ values from \eqref{eqn:p_estimate_uni} and \eqref{eqn:p_estimate_rand}.}
 \label{fig:p_vs_alpha}
\end{figure} 

\subsubsection{Choosing \texorpdfstring{$p$}{TEXT} based on random sampling}
We propose a method to choose $p$,  that is valid for arbitrary graphs under a random sampling strategy. In random sampling,  a node is added to the sampling set with probability $\alpha$. 
%

We  study the error of the  approximately reconstructed signal $\hat{\fv}^{(p)}$ from \eqref{eqn:approx_fhat} with respect to  $\fv$. We show that part of this error depends only on $p$ and $\alpha$ (\autoref{lemma:error_p}). Later we show that by setting $p$ as a function of $\alpha$, we can control this error to any given tolerance (\autoref{thrm:selectin_p}), and thus a good trade-off between complexity (small $p$) and reconstruction accuracy can be achieved. Based on these results, we propose using
\begin{equation}
    \label{eqn:p_estimate_rand}
    p_{\text{rand}} =\ceil*{-\frac{1}{2\log(1- \alpha)}}.
\end{equation}
In the rest of this sub-section, we describe the mathematical derivation used to arrive at \eqref{eqn:p_estimate_rand} (see Appendix for proofs).
\begin{lemma}\label{lemma:error_p}
    If each node $i \in \Vc$ is assigned independently to the sampling set $\Sc$ with probability $\alpha$, then 
    \begin{align}
        &\EE\left[\frac{\Vert \hat{\fv}^{(p)} - \fv\Vert}{\Vert \fv \Vert} \right] \leq \EE\Vert \Pm_{\Rc}\Zm\Vert^p + \\
        &\left(1+\left(\EE\Vert \Pm_{\Rc}\Zm\Vert^{2p}\right)^{1/2}\right)\left(\EE\left[\frac{\Vert \hat{\fv} - \fv\Vert^2}{\Vert \fv \Vert^2} \right]\right)^{1/2},
    \end{align}
    where the expectation is taken with respect to $\Sc$.
\end{lemma}
The first term depends on $p$ and $\alpha$ and the norm $\Vert \Pm_{\Rc}\Zm\Vert^p$ converges to zero as $p\rightarrow \infty$ (due to \autoref{lemma:non_expansive}). In addition since $\lim_{p \rightarrow \infty}\hat{\fv}^{(p)} = \hat{\fv}$, the second term is dominated by $\EE\left[{\Vert \hat{\fv} - \fv\Vert^2}/{\Vert \fv \Vert^2} \right]$, which is independent of $p$ and represents how good the signal reconstruction with \eqref{eqn:recon_soln} is, on average, for random sampling sets of expected size $\EE[\vert \Sc \vert] = \alpha N$.  
 The trade-off between $p$ and the sampling rate $\alpha$ is characterized by $\EE[\Vert \Pm_{\Rc}\Zm \Vert^p]$ in the next theorem.
\begin{theorem}
\label{thrm:selectin_p}
    If $p \geq 2$ and the sampling rate is $\alpha$, then
    \begin{equation}
        \sqrt{1-\alpha} \leq \left(\EE\Vert \Pm_{\Rc}\Zm\Vert^p \right)^{\frac{1}{p}}\leq \sqrt{1-\alpha} + \sqrt{\sigma_{\frac{p}{2}}(\alpha)},
    \end{equation}
    where 
    \begin{equation}
     \sigma_{\frac{p}{2}}^{\frac{p}{2}}(\alpha)   = \EE\Vert \Zm^{\top}\Pm_{\Rc}\Zm- (1-\alpha)\Zm^{\top}\Zm \Vert^{\frac{p}{2}}
    \end{equation}
    is the centered moment of order $p/2$, since $\EE[\Pm_{\Rc}] = (1-\alpha)\Id$.
\end{theorem}
This result shows that $\EE\Vert \Pm_{\Rc}\Zm\Vert^p$ concentrates around $\left(\sqrt{1-\alpha}\right)^p$. We choose $p$ as the smallest value that  ensures $\left(\sqrt{1-\alpha}\right)^p \leq \delta$ for some $\delta>0$. This can be obtained if
\begin{equation}\label{eqn:p_upper_bound_rand}
    p \geq  {2\log({\delta})}/{\log(1 - \alpha)}.
\end{equation}
Our proposed practical rule in \eqref{eqn:p_estimate_rand} is obtained after setting $\log(1/\delta)=1/4$ and rounding up  \eqref{eqn:p_upper_bound_rand} to an integer.
The choice $\log(1/\delta)=1/4$ is justified by  following inequality due to Taylor's series approximation
\begin{equation}\label{eq_taylor_approx_p}
    -\frac{1}{2\log(1- \alpha)} =  \frac{1}{2 (\alpha + \frac{\alpha^2}{2} +\frac{\alpha^3}{3}+ \cdots)} \leq \frac{1}{2 \alpha},
\end{equation}
which  implies that for any $\alpha$ we have
\begin{equation}
    p_{\text{rand}} \leq p_{\text{uni}}.
\end{equation}
The comparison between  \eqref{eqn:p_estimate_uni} and \eqref{eqn:p_estimate_rand} for different sampling rates $\alpha$ is shown in \autoref{fig:p_vs_alpha}.  
When $\alpha$ is high ($0.3 - 0.5$),   $p$ will be $1$ or $2$, and the interpolators will be extremely localized and sparse, which can be exploited for efficient inner product computation in \hyperref[algo:fast_sampling]{RAGS}, thus allowing us to sample efficiently at higher rates.

In the next section, we introduce our block-based sampling algorithm. We also include a complexity analysis of both algorithms, which illustrates the effect of $p$ on complexity.

\section{Block-based PC subsampling using sub-graphs with self-loops}
\label{sec:block_pcsampl}

 So far, we have developed a sampling algorithm (\hyperref[algo:fast_sampling]{RAGS}) that acts on the whole point cloud. As will be discussed in \autoref{subsec:complexity_analysis}, the complexity of computing the interpolating vectors and sampling set selection scale significantly with the number of points $N$, even if graph sparsity is exploited. 
 We propose to employ a block-based sampling approach, reconstruction-aware block sampling (\hyperref[algo:block_sampling]{RABS}) to further reduce the complexity of \hyperref[algo:fast_sampling]{RAGS} when $N$ is large (e.g., $N > 10^6$).In  We show that  localizing the interpolating vectors to each block,  
 reduces their dimension, 
 and thus the overall complexity. 
 By performing sampling independently in each block, we also reduce the number of comparisons. 
Note that during block partitioning, the graph edges across the block boundaries are removed. 
To account for this loss of information while computing the block-wise interpolators, we propose adding self-loops (node weights) to the points along block boundaries. 
 
 \subsection{Problem formulation}
 \label{subsec:block_problem_forml}
We wish to obtain local sampling sets $\{\Sc_{j}\}_{j=1}^{L}$ on subgraphs $\{\Gc_{j}\}_{j=1}^{L}$ such that $\bigcup_{j = 1}^{L}\Sc_{j}$ is as close as possible to the global sampling set $\Sc^{*}$ obtained for the entire graph $\Gc$. 
The point set $\Pc$ is partitioned into $L$ non overlapping blocks (subsets of points) denoted by $\Bc_{1}, \Bc_{2}, \dots \Bc_{L}$ such that $\Pc = \bigcup_{j=1}^{L} \mathcal{B}_{j}$, with $\abs{\mathcal{B}_{j}} = b_{j},  \forall j$. The corresponding vertex sets are denoted by $\Vc_{1}, \Vc_{2}, \dots \Vc_{L}$. We sample $s_i$ points from the $i$th block such that $s_i/b_i \leq \alpha$, thus keeping the global sampling rate below $\alpha$. The parameter $p$ is the same for all blocks.
 
 We first consider the case of two blocks $\mathcal{B}_{1}$ and $\mathcal{B}_{2}$ and show the modifications that need to be made to \hyperref[algo:fast_sampling]{RAGS} for block-wise sampling. At the end of this section, we extend our solution to $L> 2$ blocks. 
 
We denote $\Qm^{\Gc}(p)$ the matrix of interpolating vectors for graph $\Gc$. 
Consider sub-graphs $\Gc_1, \Gc_2$, with corresponding interpolator matrices $\Qm^{\Gc_1}(p)$ and $\Qm^{\Gc_2}(p)$, which will be explicitely determined later in this section. Then, rewriting the objective \eqref{eqn:objective_function} for $m \in \{1, 2\}$, $\Sc_m$ is the solution to 
\begin{align}\label{eq_optimization_2blocks}
    \max_{\Sc_{m} : \vert \Sc_{m} \vert \leq s_{m}} \det \left((\Qm^{\Gc_m}(p))_{\Vc_m \Sc_m}^{\top} (\Qm^{\Gc_m}(p))_{\Vc_m \Sc_m}\right). 
\end{align}
The number of samples $s_{m}$ is chosen so that $s_m/b_m \leq \alpha$ and   $s_{1}+s_{2}=s$. 
%
While independent sampling using subgraphs $\Gc_{1}$ and $\Gc_{2}$ can reduce overall computational complexity, in general,  $\Sc_{1} \cup \Sc_{2} \neq \Sc^{*}$ if the interpolators and sub-graphs are computed naively. In particular, if the subgraphs are obtained by removing connections across block boundaries, \eqref{eq_optimization_2blocks} tends to oversample points near the block boundaries. A comparison of global sampling with this naive approach is depicted in  \autoref{subfig:global_sampling} and \autoref{subfig:local_sampling}.

\begin{figure*}[t]
     \centering
        \begin{subfigure}[b]{0.30\textwidth}
            \centering
            \includegraphics[width=\linewidth]{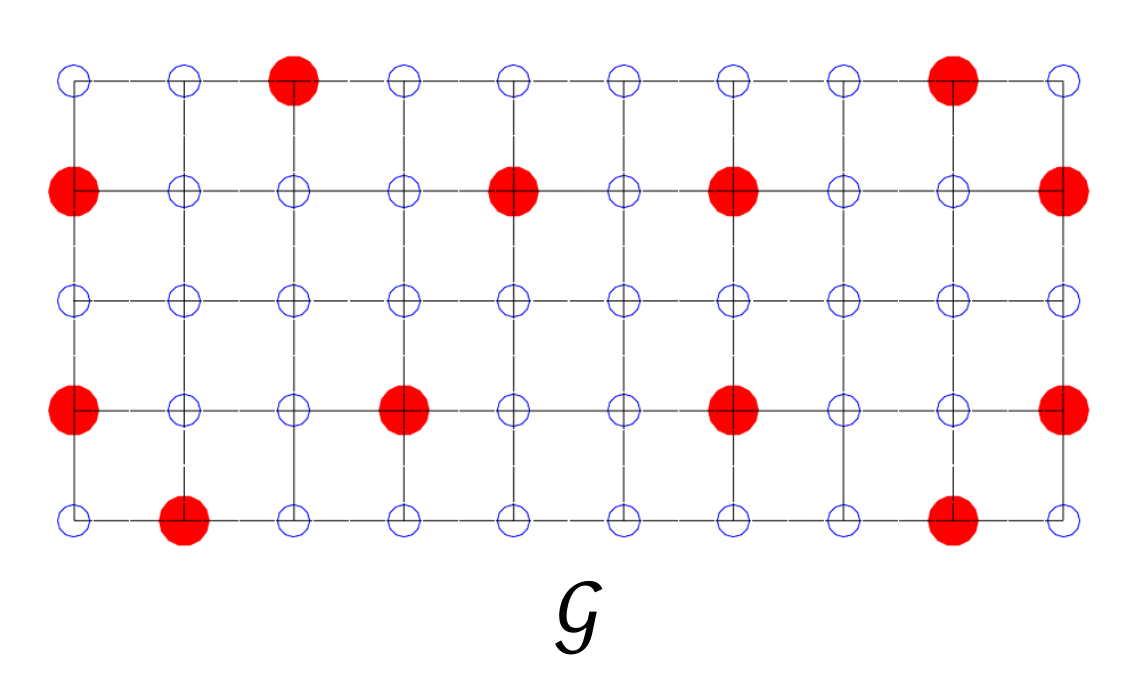}
            \caption{$\Sc^{*}$}
            \label{subfig:global_sampling}
        \end{subfigure}
     \begin{subfigure}[b]{0.30\textwidth}
         \centering
         \includegraphics[width=\linewidth]{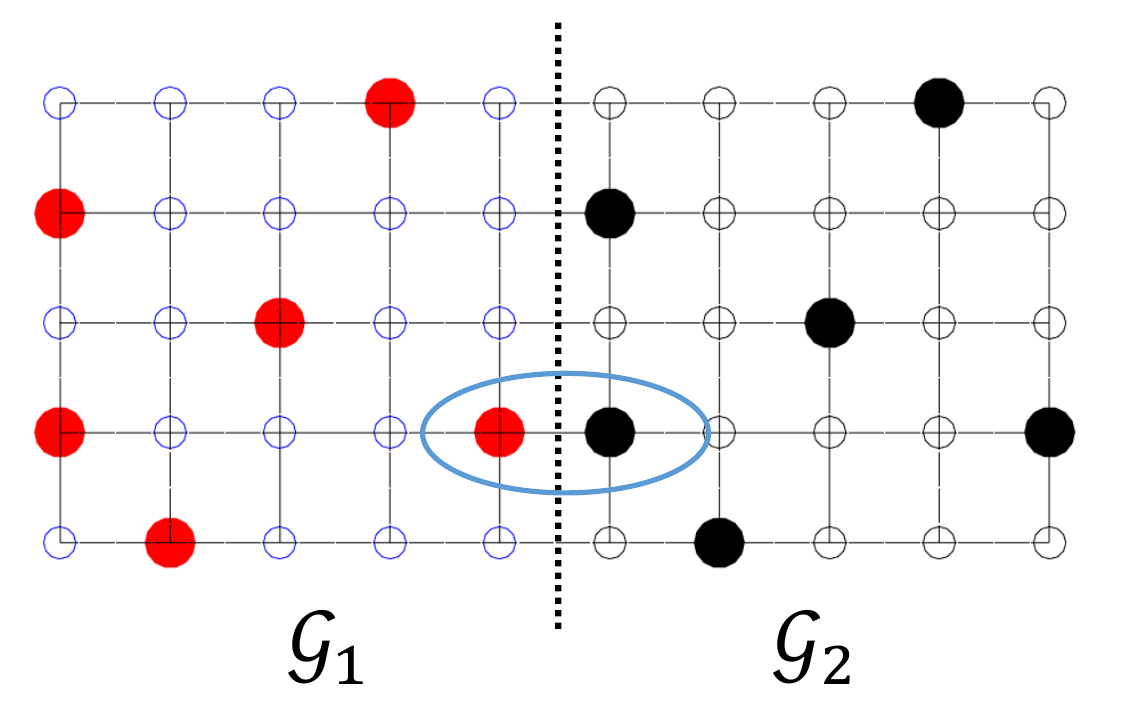}
         \caption{$\Sc_{1} \cup \Sc_{2}$}
         \label{subfig:local_sampling}
     \end{subfigure}
      \begin{subfigure}[b]{0.30\textwidth}
         \centering
         \includegraphics[width=\linewidth]{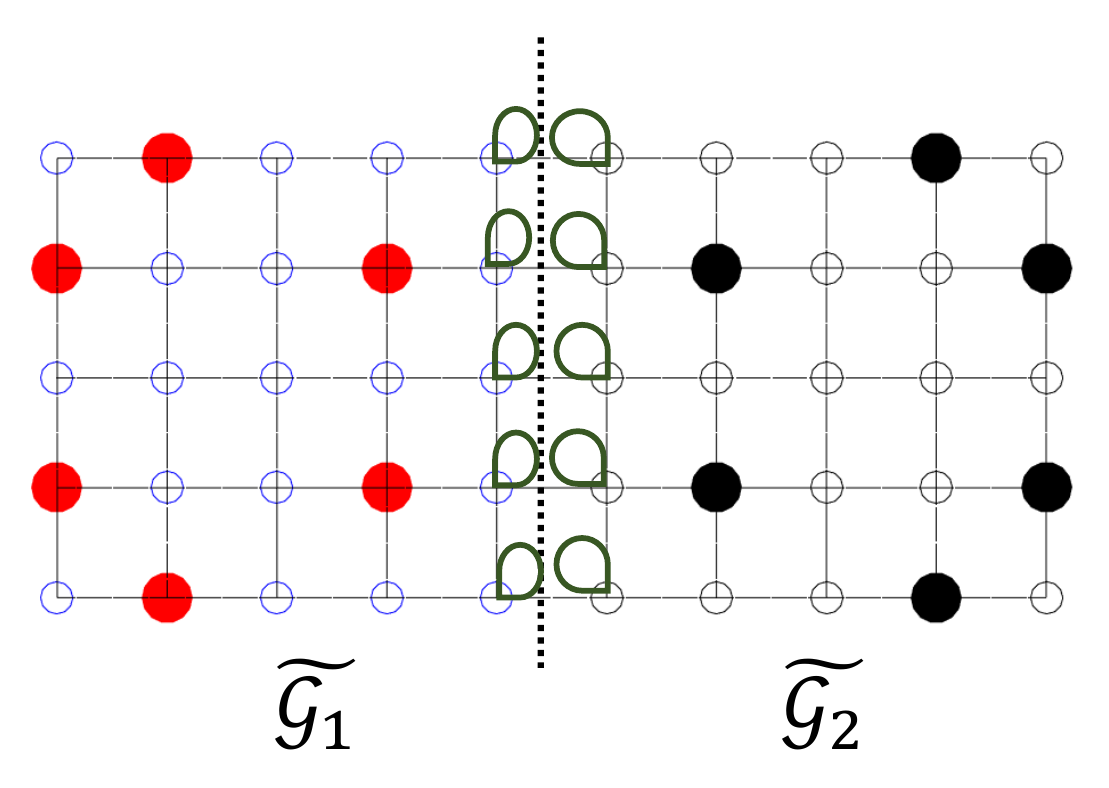}
         \caption{$\Tilde{\Sc_{1}} \cup \Tilde{\Sc_{2}}$}
         \label{subfig:local_sampling_self}
     \end{subfigure}
\caption{Comparison between sampling set obtained when the original graph $\mathcal{G}$ is sampled, sampling set obtained by sampling sub-graphs independently vs sampling set obtained by sampling modified sub-graphs with self-loops independently. Sampling sets in all three cases are obtained with $p = 3$.}
\label{fig:global_vs_local}
\end{figure*}

\subsection{Computing local interpolating vectors on sub-graphs}

A solution to the poor sampling caused by block-wise graphs would be to truncate the interpolators computed based on $\Gc$ so that they only have values in $\Gc_m$,  i.e., 
\begin{equation}
\label{eqn:local_global_intpl}
    \begin{aligned}
         \Qm^{\Gc_m}(p) = (\Qm^{\Gc}(p))_{\Vc_m \Vc_m},
    \end{aligned}
\end{equation}
for $m=1,2$. 
However, this still requires computing the global interpolating vectors, which is computationally expensive. 

Instead, we propose to identify modified block-wise graphs, $\Tilde{\Gc_{1}}$, whose local interpolators approximate the truncated interpolator of  \eqref{eqn:local_global_intpl}, i.e., such that 
\begin{equation}
\label{eqn:modified_graph_opt}
     \min_{\Tilde{\Gc}_{1}}\Vert  \Qm^{\Tilde{\Gc}_{1}}(p) -(\Qm^{\Gc}(p))_{\Vc_1 \Vc_1} \Vert_F^2.
\end{equation}
From \eqref{eqn:interpolator_def}, we can equivalently write  \eqref{eqn:modified_graph_opt} in-terms of  $\Zm$ as
\begin{align}
\label{eqn:frob_norm_obj}
   J(\Tilde{\Zm}_{\Tilde{\Gc}_{1}}) = 
   \bignorm{ \sum_{l = 1}^{p} \left(\begin{bmatrix}
            \Zm_{\Vc_{1} \Vc_{1}} &  \Zm_{\Vc_{1} \Vc_{2}}\\
            \Zm_{\Vc_{2} \Vc_{1}} &  \Zm_{\Vc_{2} \Vc{2}}
        \end{bmatrix}^{l}\right)_{\Vc_1 \Vc_1} - \sum_{l = 1}^{p} (\Tilde{\Zm}_{\Tilde{\Gc}_{1}})^{l}}_{F}^{2}.
\end{align}
When $p = 1$,  it is easy to notice that minimizing
 \eqref{eqn:frob_norm_obj} with respect to graph operator $\Tilde{\Zm}_{\Tilde{\Gc}_{1}}$ for the graph $\Tilde{\Gc}_{1}$ will result in,
\begin{equation*}
    \Tilde{\Zm}_{\Tilde{\Gc}_{1}}= \Zm_{\Vc_{1} \Vc_{1}} 
\end{equation*}
But, $\Zm_{\Vc_{1} \Vc_{1}} = \frac{1}{2}(\Id + \Dm_{\Vc_{1} \Vc_{1}}^{-1} \Am_{\Vc_{1} \Vc_{1}})$, where $\Am_{\Vc_{1} \Vc_{1}}$ is the adjacency matrix of subgraph $\Gc_{1}$. The degree corresponding to the points $i \in \Vc_{1}$ according to $\Dm_{\Vc_{1} \Vc_{1}}$ can be written as,
\begin{equation}
    [\Dm_{\Vc_{1} \Vc_{1}} ]_{ii}= 
\begin{dcases}
    \sum_{j \in \Vc_{1}} w_{ij} + \sum_{k \in \Vc_{2}} w_{ik}, & \text{if }  \Nc(i) \subset \Vc_{2} \\
    \sum_{j \in \Vc_{1}} w_{ij},              & \text{otherwise},
\end{dcases}
\end{equation}
where $\sum_{k \in \Vc_{2}} w_{ik}$ is the node weight (self-loop weight) on the point that shares edges with the neighboring blocks.
Therefore, the optimal sub-graph $\Tilde{\Gc}_{1}$ when $p=1$ is the sub-graph $\Gc_{1}$ with node set $\Vc_1$ and weighted self-loops (node weights) given by,
\begin{equation}
\label{eqn:self_loop_weight}
    \Phim_{ii}  = \sum_{k \in \Vc_{2}} w_{ik}.
\end{equation}
We can write the interpolating vectors for $p=1$ with respect to the optimized sub-graph $\Tilde{\Gc}_{1}$ as,
\begin{equation}
    \label{eqn:interpolating_vec_p1}
   \Qm^{\Tilde{\Gc}_{1}}(1) = \frac{1}{2}(\Id + (\Dm_{\Gc_{1}} + \Phim_{\Gc_{1}})^{-1} \Am_{\Gc_{1}}),  
\end{equation}
where $\Am_{\Gc_{1}} = \Am_{\Vc_{1} \Vc_{1}}$ and $\Dm_{\Gc_{1}}$ are the adjacency matrix and degree matrix of sub-graph $\Gc_{1}$ respectively. $\Phim_{\Gc_{1}}$ is the diagonal self-loop matrix with weights given in \eqref{eqn:self_loop_weight}. Thus, when $p=1$, minimizing the objective in \eqref{eqn:frob_norm_obj} does not require computing the interpolating vectors of the entire graph $\Gc$. Instead, we can compute vectors independently on the sub-graphs of $\Gc$ using \eqref{eqn:interpolating_vec_p1} while satisfying \eqref{eqn:local_global_intpl}. 
However, for $p > 1$, the truncation in \eqref{eqn:local_global_intpl} requires calculating interpolating vectors of the entire graph $\Gc$. To circumvent this, we  obtain the sub-graph interpolators recursively from \eqref{eqn:interpolating_vec_p1} using
\begin{equation}
\label{eqn:interp_vec_p}
    \Qm^{\Tilde{\Gc}_{1}}(p) =  \Zm_{\Vc_{1} \Vc_{1}}\Qm^{\Tilde{\Gc}_{1}}(p-1) + \Qm^{\Tilde{\Gc}_{1}}(1) = \sum_{l=1}^{p} (\Zm_{\Vc_{1} \Vc_{1}})^{l}, 
\end{equation}
 where $ \Qm^{\Tilde{\Gc}_{1}}(1) = \Zm_{\Vc_{1} \Vc_{1}}$
 from \eqref{eqn:interpolating_vec_p1}. It is important to note that while \eqref{eqn:interp_vec_p} is not an optimal solution to \eqref{eqn:modified_graph_opt} for $p>1$, it is a computationally efficient approximation.  The same steps can be repeated to obtain interpolators for the sub-graph on $\Vc_2$.

The effect of self-loops in block-based sampling is demonstrated through a toy example in \autoref{fig:global_vs_local}. \autoref{subfig:global_sampling}  shows  $\Sc^{*}$ obtained from sampling the entire graph using RAGS with $p=3$. \autoref{subfig:local_sampling} shows the results of sampling sub-graphs independently without self-loops while
\autoref{subfig:local_sampling_self} shows the results of proposed RABS with self-loops. We can observe that $\Tilde{\Sc}_{1} \cup \Tilde{\Sc}_{2}$ is closer than $\Sc_{1} \cup \Sc_{2}$ to $\Sc^{*}$.

Next, we extend our approach to   $L>2$ blocks by leveraging PC partitioning with octrees \cite{jackins1980_octtrees}. 
Initially, a $\Pc$ is divided across the $x$, $y$, and $z$ axes,  
resulting in a total of $8$ subsets (PC blocks) after one level of octree partitioning (refer to \cite{pavez2018_polygoncloudcompr}).  This process is repeated recursively for non-empty subsets until the desired resolution level (or block size) is reached.
If there are $L$ blocks after the octree sub-division, then the self-loop weights for points in block $\Bc_{m}$ can be calculated by simply considering the edges across the block boundaries, i.e.,
$\Phim_{ii} = \sum_{k \in \Nc(\Bc_{m})} w_{ik}$, where $i \in \Bc_m$ and $\Nc(\Bc_{m})$ are the points in the neighboring blocks. Therefore, for each block, we find the boundary points that share edges with the points in neighboring blocks and compute the edge weight for each point using \eqref{eqn:self_loop_weight}. We can obtain self-loop matrix for the block $\Bc_{m}$ by a sparse matrix-vector product given by,

\begin{equation}
    \label{eqn:self_loop_matrix}
    \Phim_{\Gc_{m}} =  \diag (\Am_{m} \cv^{(m)}),
\end{equation}
where $\Am_{m} \in \mathbb{R}^{b_{m} \times N}$ is constructed by keeping the rows indexed by $\Vc_{m}$ of the adjacency matrix of  $\Gc$. $\cv^{(m)} \in \mathbb{R}^{N}$ is a vector defined as $\cv^{(m)}_{i} = 1, \forall i \notin \Vc_{m}$ and zero otherwise.

\subsection{Complexity }
\label{subsec:complexity_analysis}
In this section, we compare the run-time complexity of the proposed \hyperref[algo:fast_sampling]{RAGS} and  \hyperref[algo:block_sampling]{RABS} algorithms with existing graph signal sampling algorithms (see \autoref{tab:complexity}).
\subsubsection{Complexity of RAGS} This algorithm has two main stages: preparation and sampling. The preparation step consists of  obtaining the interpolators $\mathbf{q}_{i}^{(p)}$ (i.e., the columns of $\Qm(p)$) and the inner products $\langle \mathbf{q}^{(p)}_{i},\mathbf{q}^{(p)}_{j} \rangle $,  $\forall i,j \in \lbrace 1, \cdots, N \rbrace$. The sampling step implements the greedy update \eqref{eqn:greedy_update}. An important quantity appearing in our complexity analysis is the maximum number of non-zero entries of each row and column of  $\Zm$, which is denoted by $\Bar{d}$. For KNN graphs $\Bar{d} = k+1$, the unweighted degree of $\Gc$ is $\Bar{d}-1$.
The complexity of the preparation step is dominated by the computation of the inner products $\langle \mathbf{q}^{(p)}_{i},\mathbf{q}^{(p)}_{j} \rangle $ resulting in a total cost of $\Oc(\Bar{d}^{3p} N)$. Once the inner products are available, the complexity of the greedy update is $\Oc(sN)$, which is a result of the number of additions and comparisons per iteration and the total number of iterations. See Appendix \ref{app:complexity} for a detailed explanation.
\subsubsection{Complexity of RABS}
\label{subsec:block_sampl_complexity} 
Because  \hyperref[algo:block_sampling]{RABS} is based on \hyperref[algo:fast_sampling]{RAGS} applied to each block independently, we can estimate its complexity using similar arguments. Since the graphs in each block are sub-graphs of the original graph, the parameter $\Bar{d}$ estimated for the whole graph can be used as an upper bound for each sub-graph. 
%
Because the interpolator vectors have lower dimensions, the complexity of the preparation stage can be reduced to $\Oc(\Bar{d}^p b_i^2)$ for the $i$th block (see \autoref{remark_app_complexity_densegraph} in the Appendix). Thus, given that  $N \leq \sum_{i=1}^L b_i^2 \leq \Bar{B} N$, where $\Bar{B} = \max_{i}b_i$ is the maximum number of points in a block, the overall complexity of interpolator and inner product computations is $\Oc(\Bar{d}^p \sum_{i=1}^Lb_i^2) = \Oc(\Bar{B}\Bar{d}^pN)$. 
%
Block partitioning using octree has complexity $\Oc(N\log(N))$, which is also part of the preparation step.  Greedy sampling of each block runs in $\Oc(s_i b_i)$ time, resulting in  total complexity of  sampling equal to $\Oc(\sum_{i=1}^L s_i b_i)$.  If $s_i = { \alpha b_i}$ and $\alpha = s/N$   we have that 
\begin{equation}
    \sum_{i=1}^L s_i b_i \leq s \max_i b_i = s\Bar{B},
\end{equation}
which can be much smaller than $\Oc(sN)$, especially when the number of points per block or the size of the sampling set is small compared to $N$. In our experiments, $N$ can be in the order of hundreds of thousands, while $\Bar{B}$ can be a few hundred points, resulting in orders of magnitude runtime reduction. 
Note that this complexity reduction with respect to global sampling is not due to parallelization. 
\begin{table}[tp]
\centering
\caption{Complexity comparison - AVM: $m$ is the order of the filter, $T_{1}$ is the number of iterations to compute cut-off frequency, BS-GDA: $\epsilon$ is the numerical precision to terminate the binary search and $P, Q$ represents the number of nodes and edges within specified hop distance from points}
\begin{tabular}{|c|c|c| c|}
\hline
\begin{tabular}[c]{@{}c@{}c@{}} Algorithm\end{tabular} & Type &
  Preparation &
  \begin{tabular}[c]{@{}c@{}c@{}}Sampling \end{tabular} \\ \hline
AVM \cite{jayawant2021_avm} & Global   & $\mathcal{O}(\abs{\mathcal{E}} m T_{1} log(N))$ &  $\mathcal{O}(s \abs{\mathcal{E}} m)$  \\ \hline
BS-GDA \cite{bai2020_bsgda} & Global & $\mathcal{O}(N (P + Q) \log_{2}\frac{1}{\epsilon})$& $\mathcal{O} (sN P  \log_{2} \frac{1}{\epsilon})$  \\ \hline
\begin{tabular}[c]{@{}c@{}c@{}}
\hyperref[algo:fast_sampling]{RAGS} \end{tabular} & Global& $\mathcal{O}( \Bar{d}^{3p}N)$
   & $\mathcal{O}(sN)$
   \\ \hline
\hyperref[algo:block_sampling]{RABS} & Block &  $\mathcal{O}(  (\Bar{B}\Bar{d}^{p}+\log(N))N)$ & $\mathcal{O}(s\Bar{B})$\\ \hline
\end{tabular}
\label{tab:complexity}
\end{table}
\subsubsection{Comparison with AVM and BS-GDA}
We now compare the complexity of the proposed  \hyperlink{\autoref{algo:fast_sampling}}{RAGS} and \hyperlink{\autoref{algo:block_sampling}}{RABS}  with the existing ED-free graph signal sampling algorithms - approximate volume maximization (AVM) \cite{jayawant2021_avm} and BS-GDA \cite{bai2020_bsgda} (see \autoref{tab:complexity}). The overall complexity is divided into two parts: 1) preparation  and 2) sampling set selection. In the sampling step, most algorithms adopt a greedy selection strategy.  

For sparse KNN graphs, since $\vert \Ec \vert = kN$, the preparation step of the four approaches scales near linearly with $N$, with different factors depending on their specific parameters. 
The main difference is observed in actual runtimes (see \autoref{sec:experiments}), where the runtime of  AVM is dominated by 
the estimation of cut-off frequencies through graph filtering, and for BS-GDA complexity is dominated by Gershgorin disk alignment of all nodes. 
In the proposed methods, complexity is dominated by inner product computation, where the complexity depends exponentially on $p$. This highlights the importance of choosing $p$ as small as possible, which is achieved using the methods from \autoref{sec:sampling_algo_development}. The dependence on $p$ is also improved by RABS  from $3p$ to $p$, due to the use of more localized interpolators. 

The sampling stage of AVM involves computing interpolators through spectral domain filtering in every iteration, and BS-GDA performs  $\mathcal{O}(\log_{2} \frac{1}{\epsilon})$  binary searches in every iteration. The sampling stage of RAGS and RABS does not involve additional computation apart from evaluating the objective function from the precomputed inner products.
While block sampling can be combined with AVM and BS-GDA, the resulting runtime would still be high because of the expensive preparation stage, which involves graph frequency analysis (see \autoref{subsec:comparison}). 

\section{Experiments}
\label{sec:experiments}
We provide a comprehensive evaluation of proposed sampling algorithms - \hyperref[algo:block_sampling]{RABS} and \hyperref[algo:fast_sampling]{RAGS} on large PCs. We first analyze the effect of the parameter $p$ on the reconstruction accuracy and run-time in \autoref{subsec:p_effect_exps}. We compare the performance of proposed RAGS and RABS algorithms and evaluate the effect of self-loops in RABS in \autoref{subsec:self_vs_noself}.
\autoref{subsec:comparison} compares against  uniform sampling \cite{sridhara2022_chromapc} and other state-of-the-art graph-based sampling methods \cite{chen2017_pcsamplgraph,jayawant2021_avm, bai2020_bsgda}. Finally, we apply our sampling algorithms to chroma subsampling in \autoref{subsec_exp_chroma}.

\begin{figure}[t]
\centering
\begin{subfigure}[b]{0.24\textwidth}
         \centering
         \includegraphics[width=\textwidth]{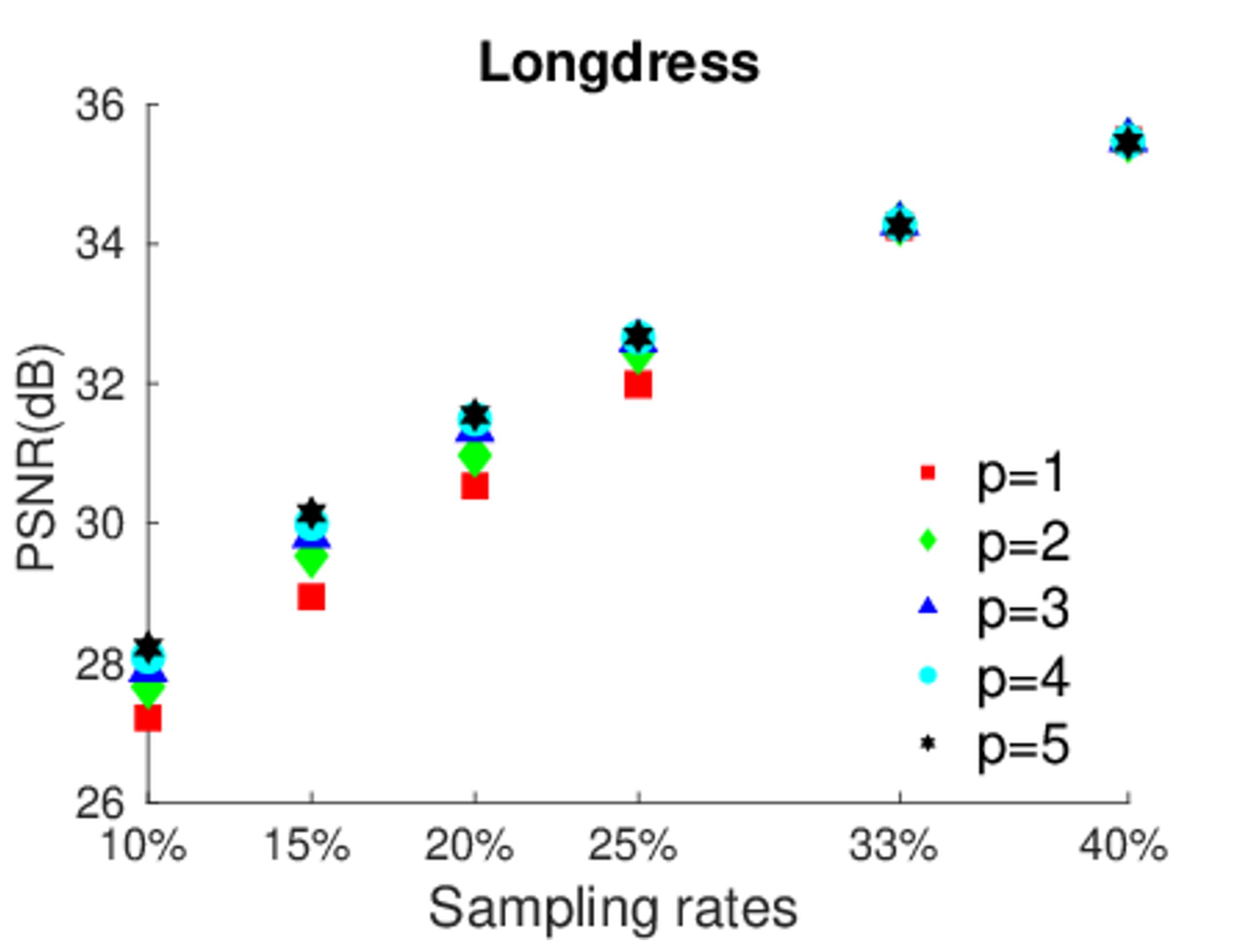}
     \end{subfigure}
     \begin{subfigure}[b]{0.24\textwidth}
         \centering
         \includegraphics[width=\textwidth]{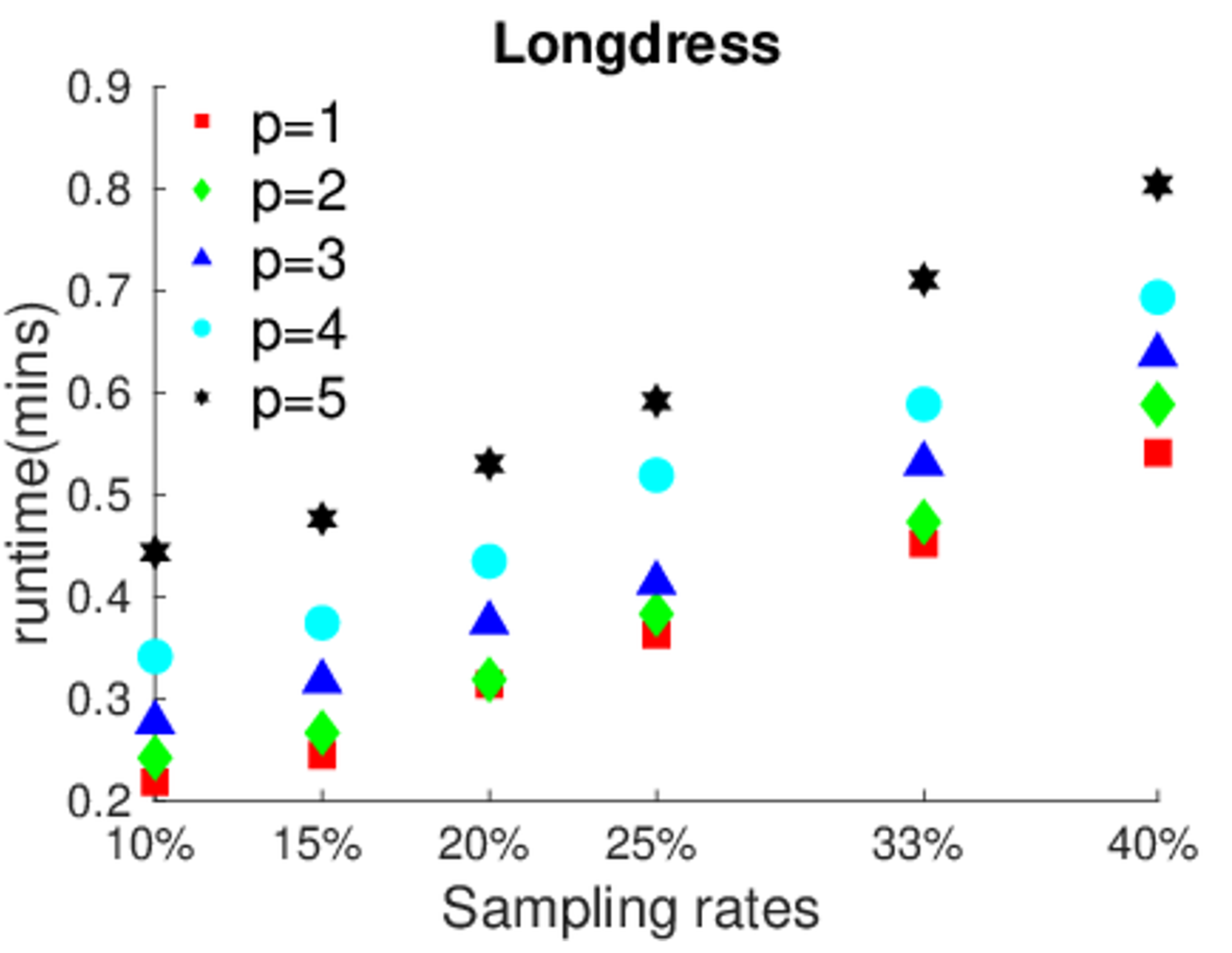}
     \end{subfigure}
     \caption{Effect of $p$ on PSNR and runtime (RABS, block-size: $2^{6} \times 2^{6} \times 2^{6}$)}
     \label{fig:p_vs_alpha_psnr}
 \begin{subfigure}[b]{0.24\textwidth}
         \centering
         \includegraphics[width=\textwidth]{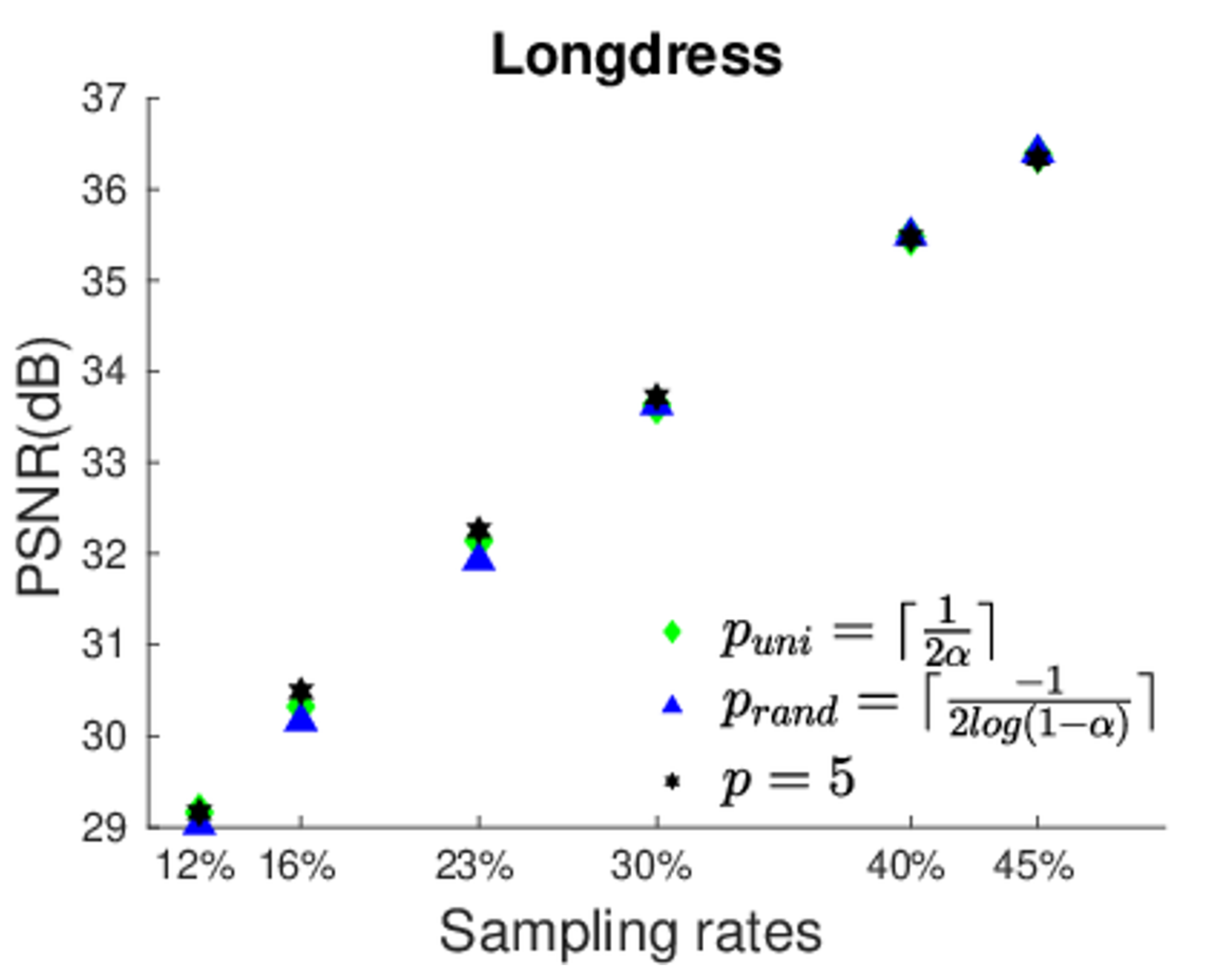}
     \end{subfigure}
\begin{subfigure}[b]{0.24\textwidth}
         \centering
         \includegraphics[width=\textwidth]{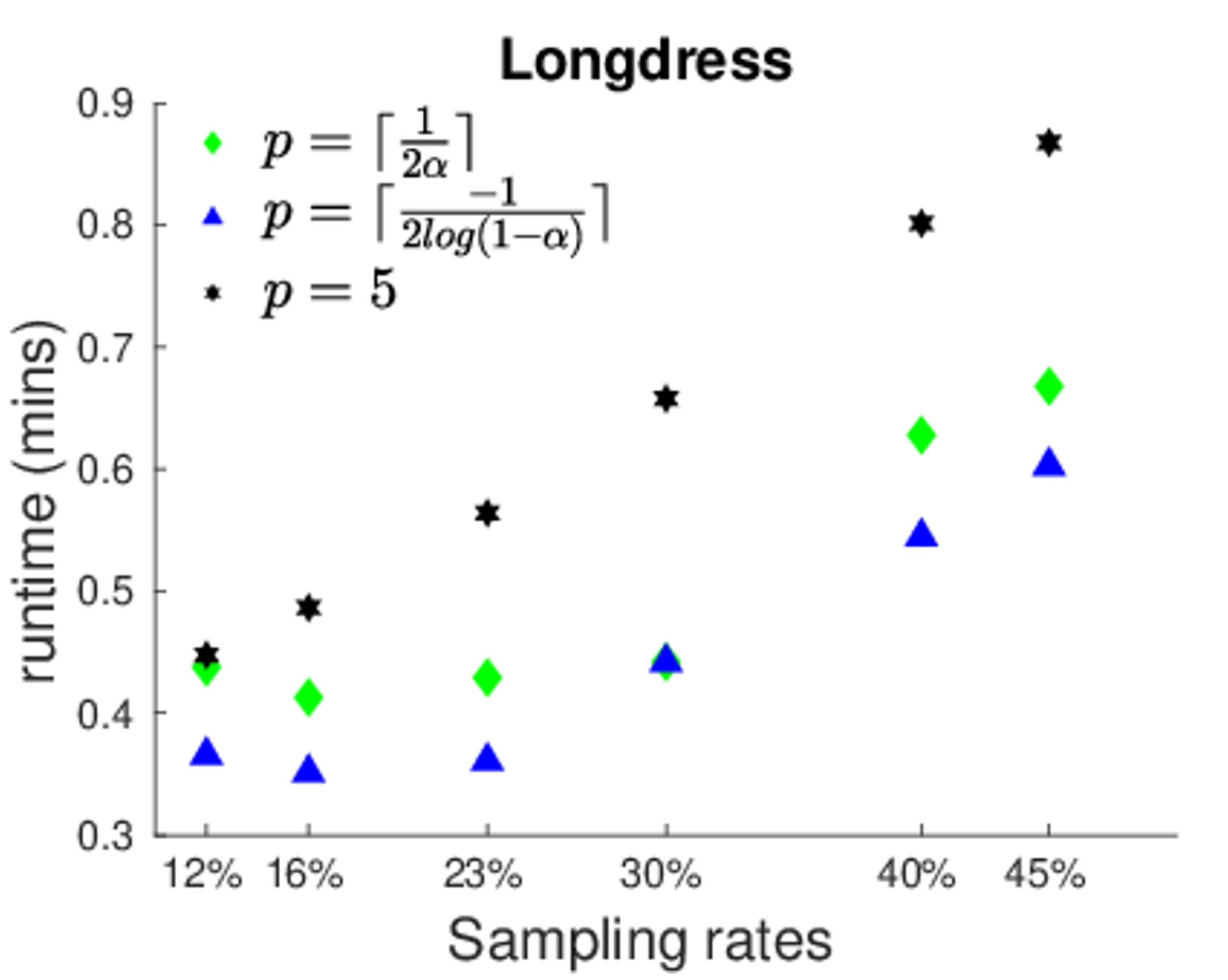}
     \end{subfigure}
\caption{Fixed $p$ vs proposed selection of $p$ based on sampling rate (RABS, block-size: $2^{6} \times 2^{6} \times 2^{6}$)}
\label{fig:fixed_p_vs_adap_p}
\end{figure}

\begin{figure}[ht]
\centering
    \begin{subfigure}[b]{0.24\textwidth}
         \centering
         \includegraphics[width=\textwidth]{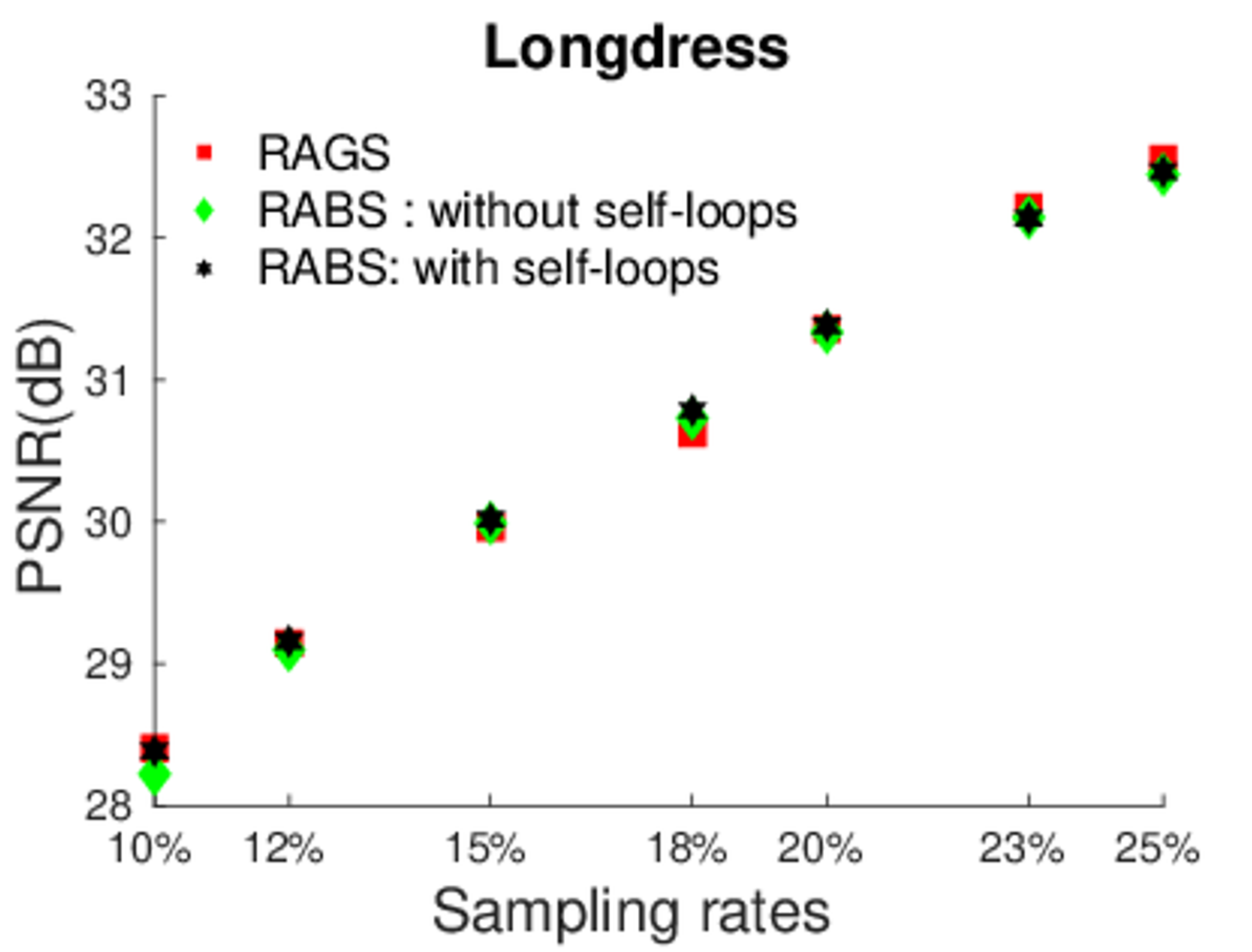}
     \end{subfigure}
     \begin{subfigure}[b]{0.24\textwidth}
         \centering
         \includegraphics[width=\textwidth]{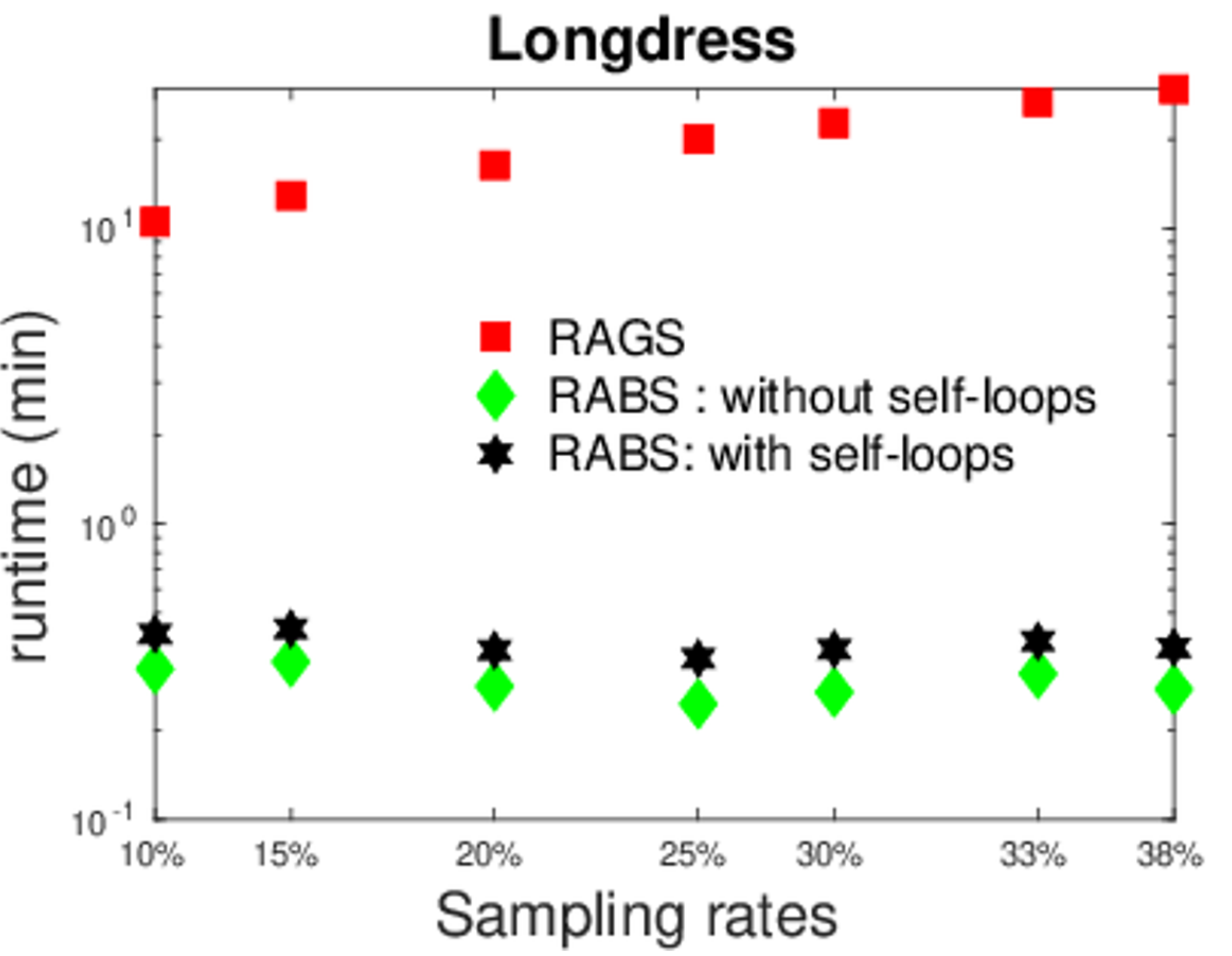}
     \end{subfigure}
     \caption{RAGS vs RABS with and without self-loops. PSNR and runtime comparison at different sampling rates for $p = \ceil{1/(2\alpha)}$, block size = $2^{6} \times 2^{6} \times 2^{6}$}
     \label{fig:global_vs_block_self}
\begin{subfigure}[b]{0.24\textwidth}
         \centering
          \includegraphics[width=\textwidth]{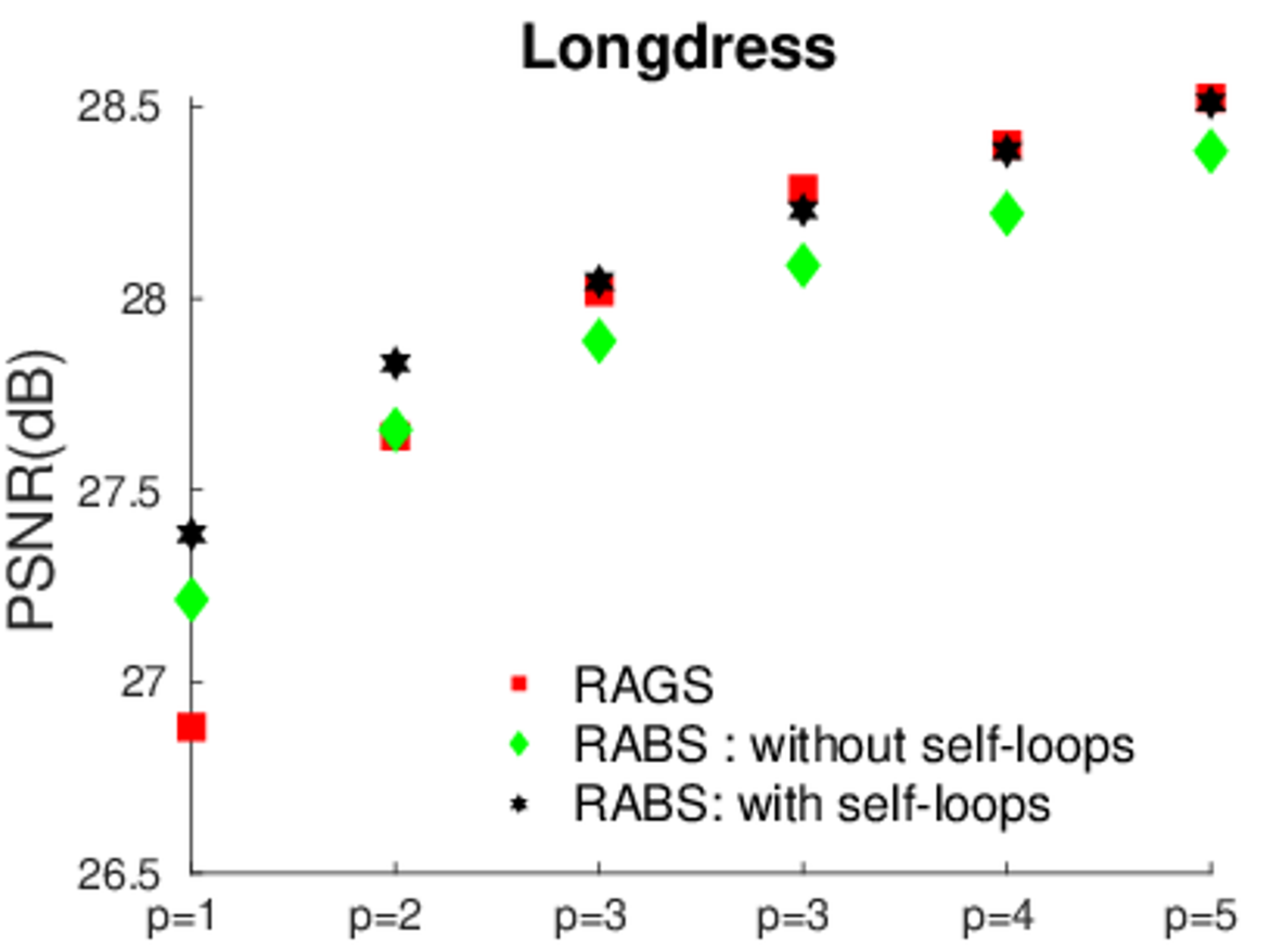}
          \caption{Sampling rate = $10\%$}
     \end{subfigure}
     \begin{subfigure}[b]{0.24\textwidth}
         \centering
         \includegraphics[width=\textwidth]{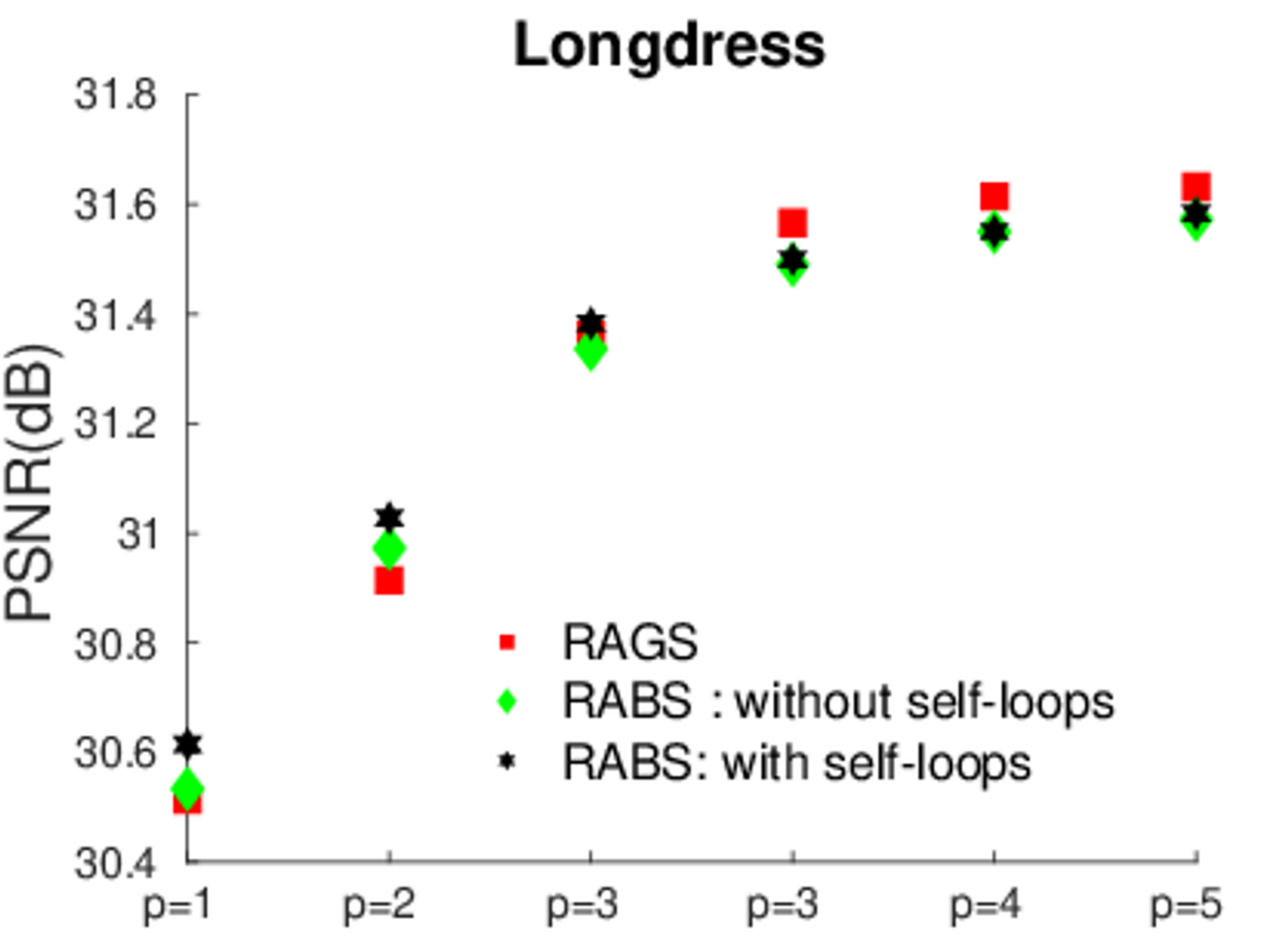}
         \caption{Sampling rate = $20\%$}
     \end{subfigure}
\caption{RAGS vs RABS with and without self-loops. PSNR comparison for different values of $p$, block size  $2^{6} \times 2^{6} \times 2^{6}$}
\label{fig:self_vs_noself}
\end{figure}

We consider two PC datasets with RGB color attributes used in the standardization activities by the motion picture experts group (MPEG)\cite{chou2017_8idataset, loop2016_mvub}. The 8i dataset \cite{chou2017_8idataset} consists of 4 PC sequences: l\textit{ongdress, redandblack, soldier, and loot}.  The Microsoft voxelized upper body (MVUB) dataset \cite{loop2016_mvub} consists of 5 PC sequences: \textit{Sarah, Andrew, David, Phil, Ricardo}. We used PCs that are voxelized to depth $10$ from both datasets. The number of points in 8i and MVUB datasets is between $700$K - $ 1.3$ million. For all our experiments, we construct a $K$-NN graph $\mathcal{G}$ as described in \autoref{sec:preliminaries} with $K = 5$ and $\sigma = (\sum_{(i, j) \in \epsilon} \norm{\mathbf{p}_{i} - \mathbf{p}_{j}}_{2}^{2})/3N$. 
Since the RGB color values are transformed to the YUV domain in most compression systems \cite{sridhara2022_chromapc, pavez2020_ragft}, we will evaluate the reconstruction accuracy on the luminance component (Y). We use the standard peak signal-to-noise ratio  (PSNR) definition for $8$-bit  color
\begin{equation}
\label{eqn:psnr_y}
         \text{PSNR}=-10\log_{10}\left( \Vert \fv-\hat{\fv}\Vert_{2}^{2}/(255^2 N)\right),
    \end{equation}
where $\fv$ and $\hat{\fv}$ are the original and reconstructed luminance signal, respectively, and $N$ is the total number of points.
We fix the block-size to $2^{6} \times 2^{6} \times 2^{6}$. We tried different block sizes and we observed larger block sizes take longer time to sample with no significant gain in reconstruction accuracy.
\begin{figure*}[t]
\centering
\begin{subfigure}[b]{0.30\textwidth}
         \centering
         \includegraphics[width=\textwidth]{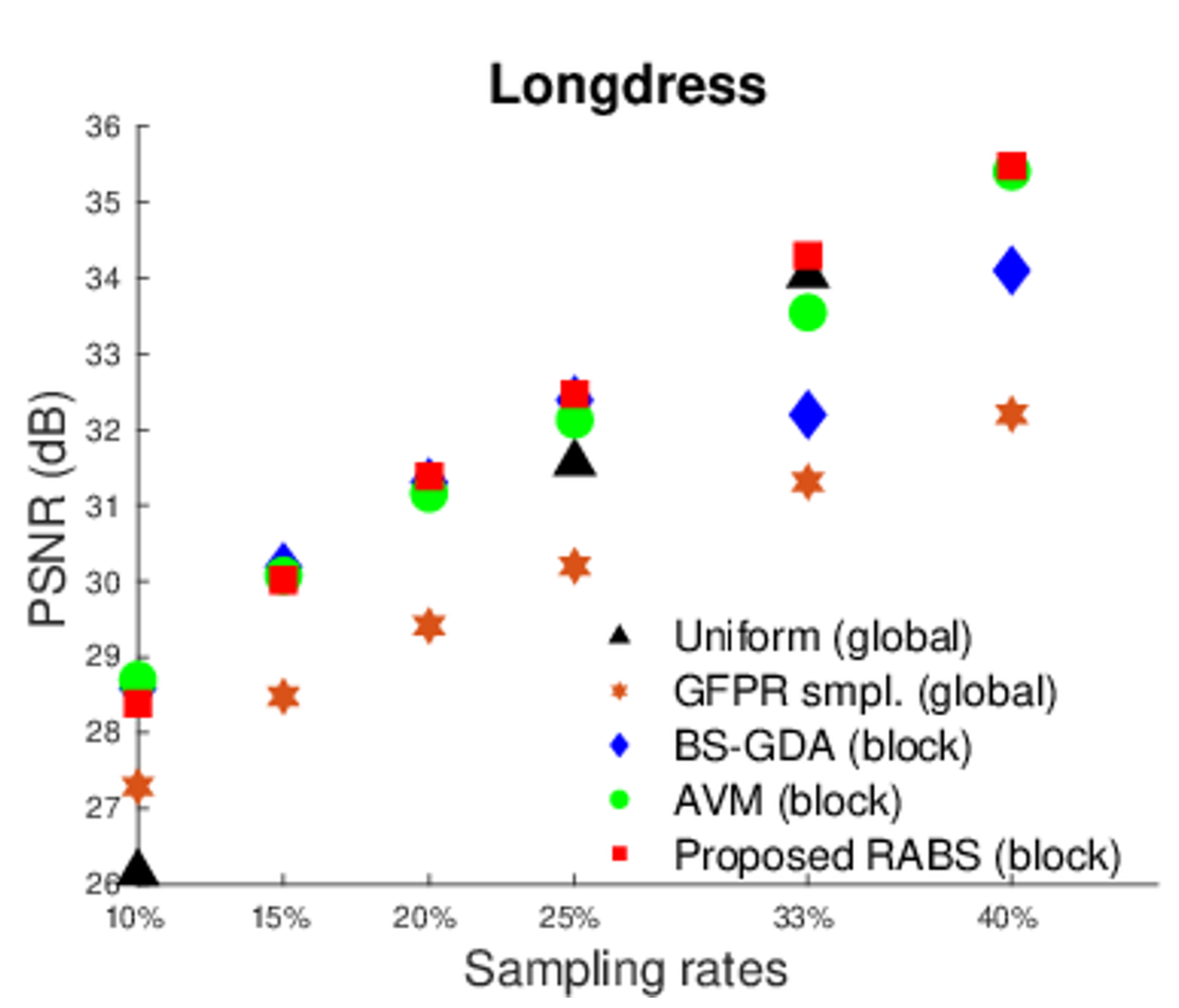}
     \end{subfigure}
     \begin{subfigure}[b]{0.30\textwidth}
         \centering
         \includegraphics[width=\textwidth]{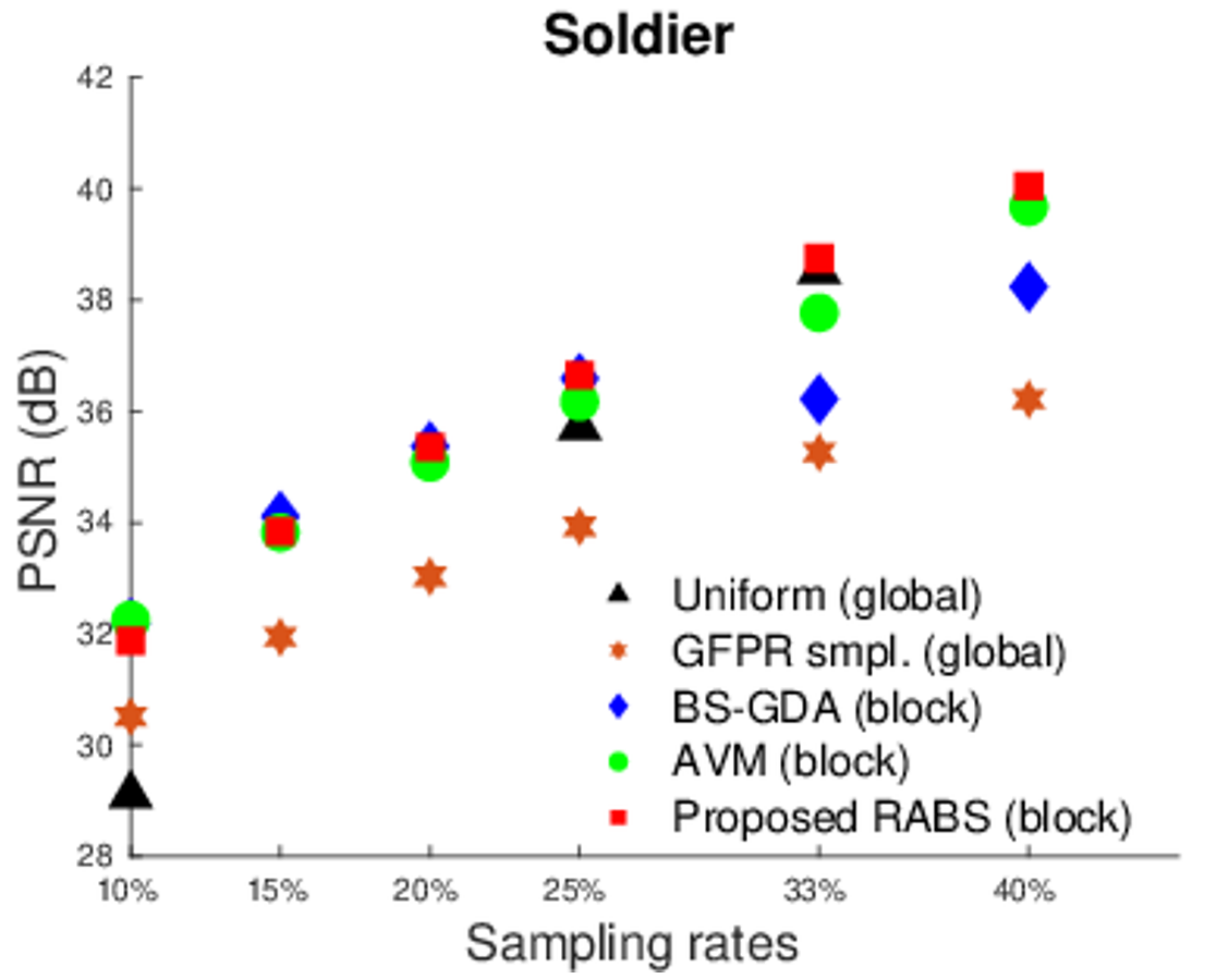}
     \end{subfigure}
     \begin{subfigure}[b]{0.30\textwidth}
         \centering
         \includegraphics[width=\textwidth]{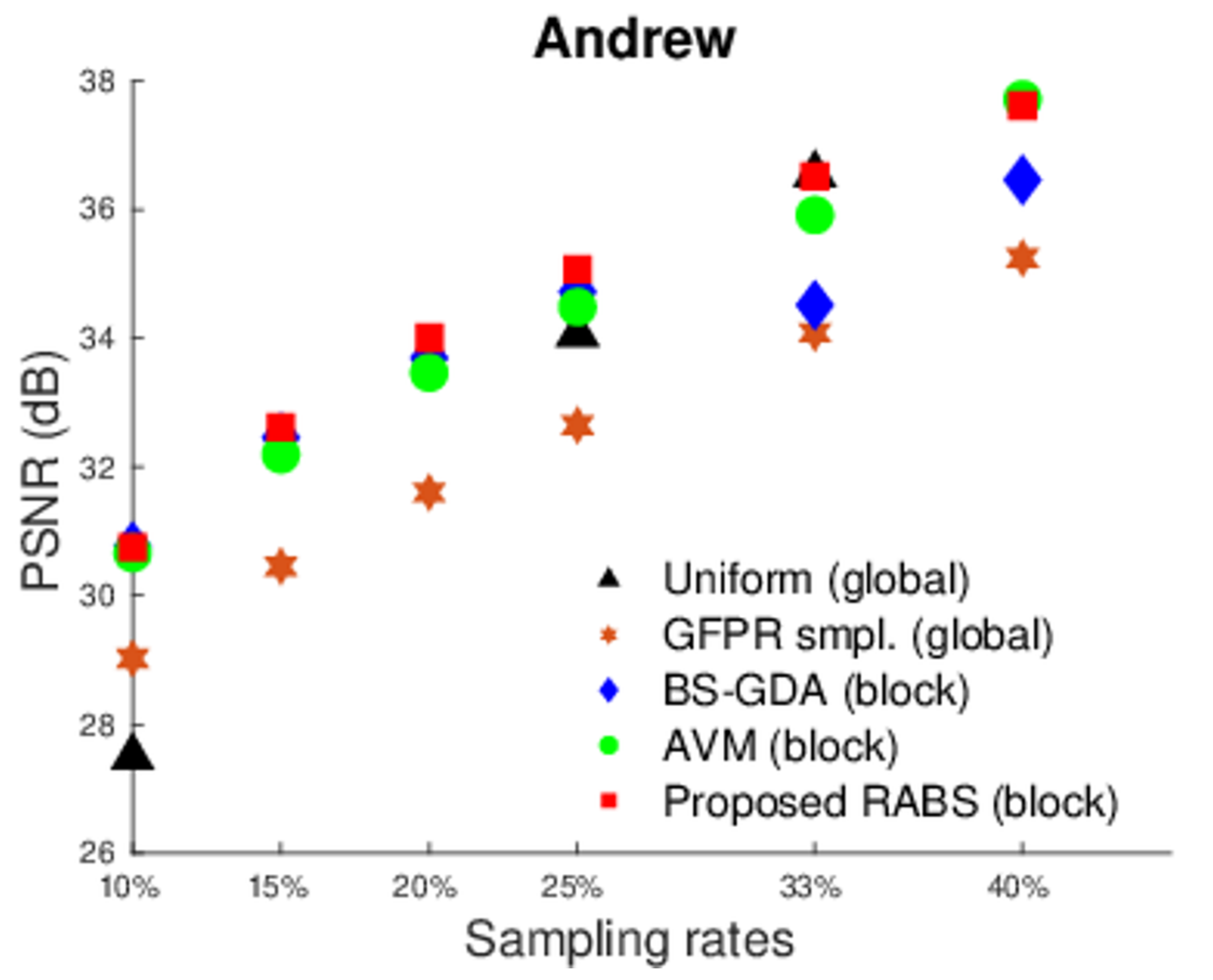}
     \end{subfigure}
\caption{PSNR comparison between different sampling algorithms. \hyperref[algo:block_sampling]{RABS} ($p = \ceil{\frac{1}{2\alpha}}$), AVM and BS-GDA are implemented in a block-based framework (block size: $2^{6} \times 2^6 \times 2^6$)}
\label{fig:psnr_comparison}
\begin{subfigure}[b]{0.30\textwidth}
         \centering
         \includegraphics[width=\textwidth]{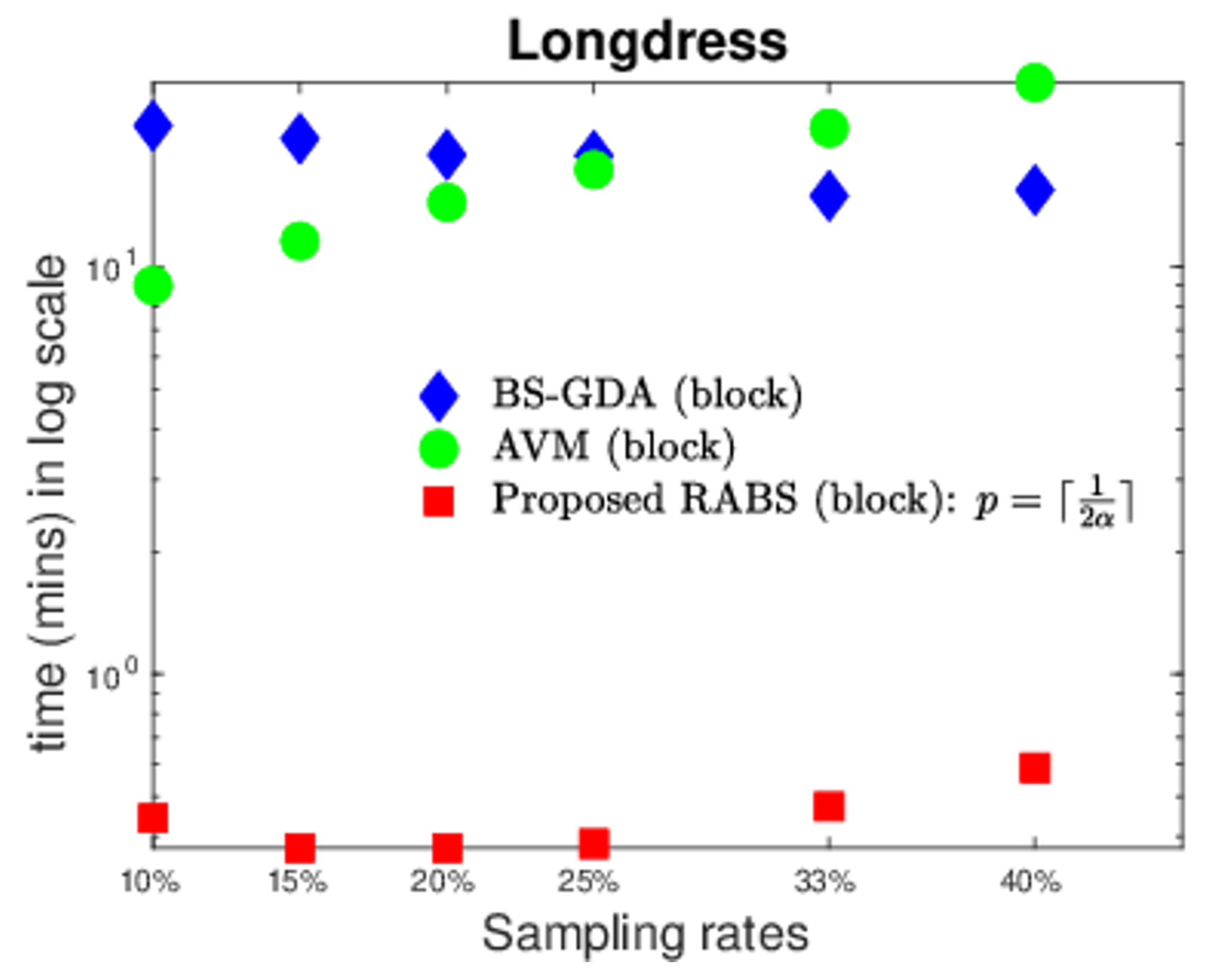}
     \end{subfigure}
     \begin{subfigure}[b]{0.30\textwidth}
         \centering
         \includegraphics[width=\textwidth]{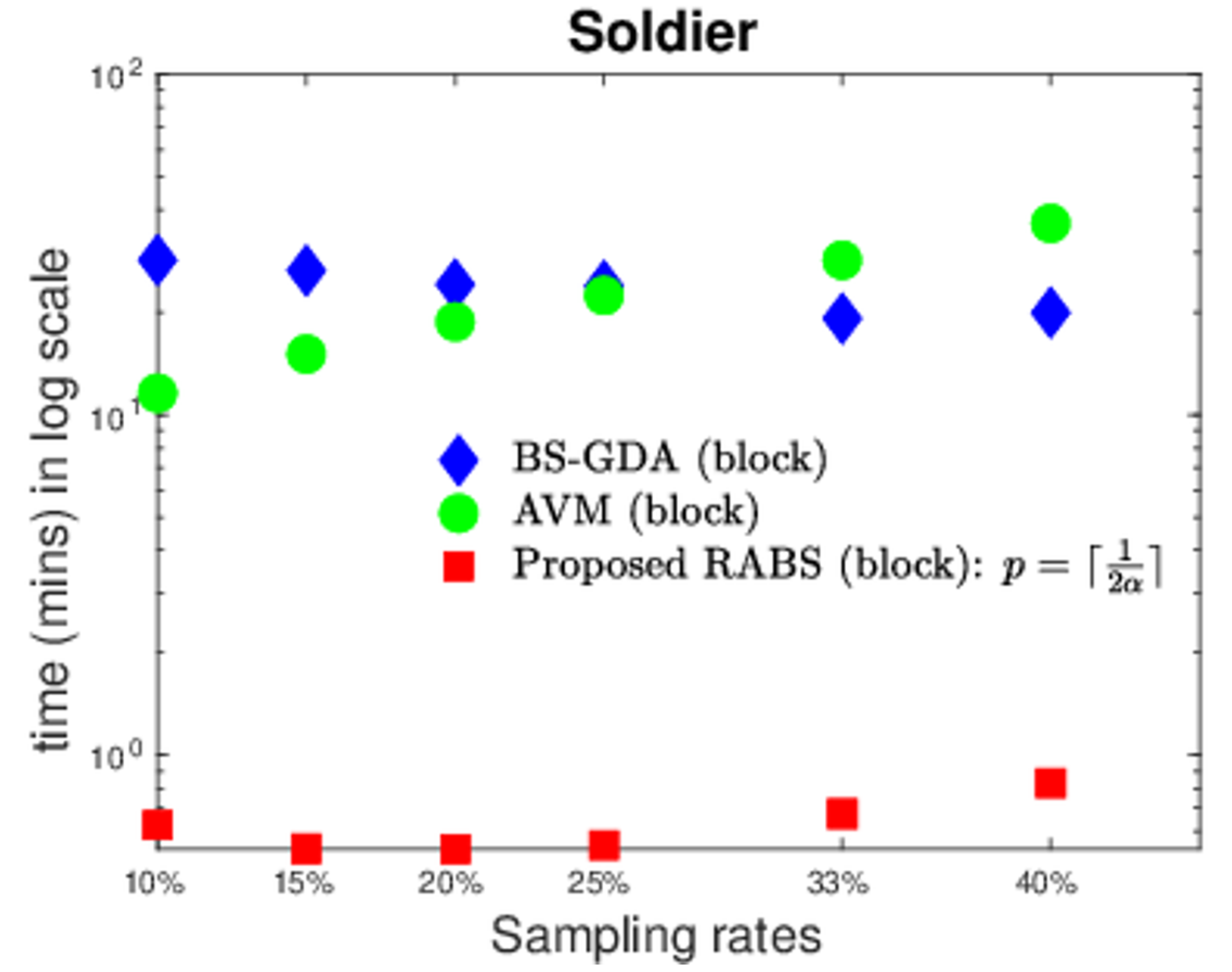}
     \end{subfigure}
     \begin{subfigure}[b]{0.30\textwidth}
         \centering
         \includegraphics[width=\textwidth]{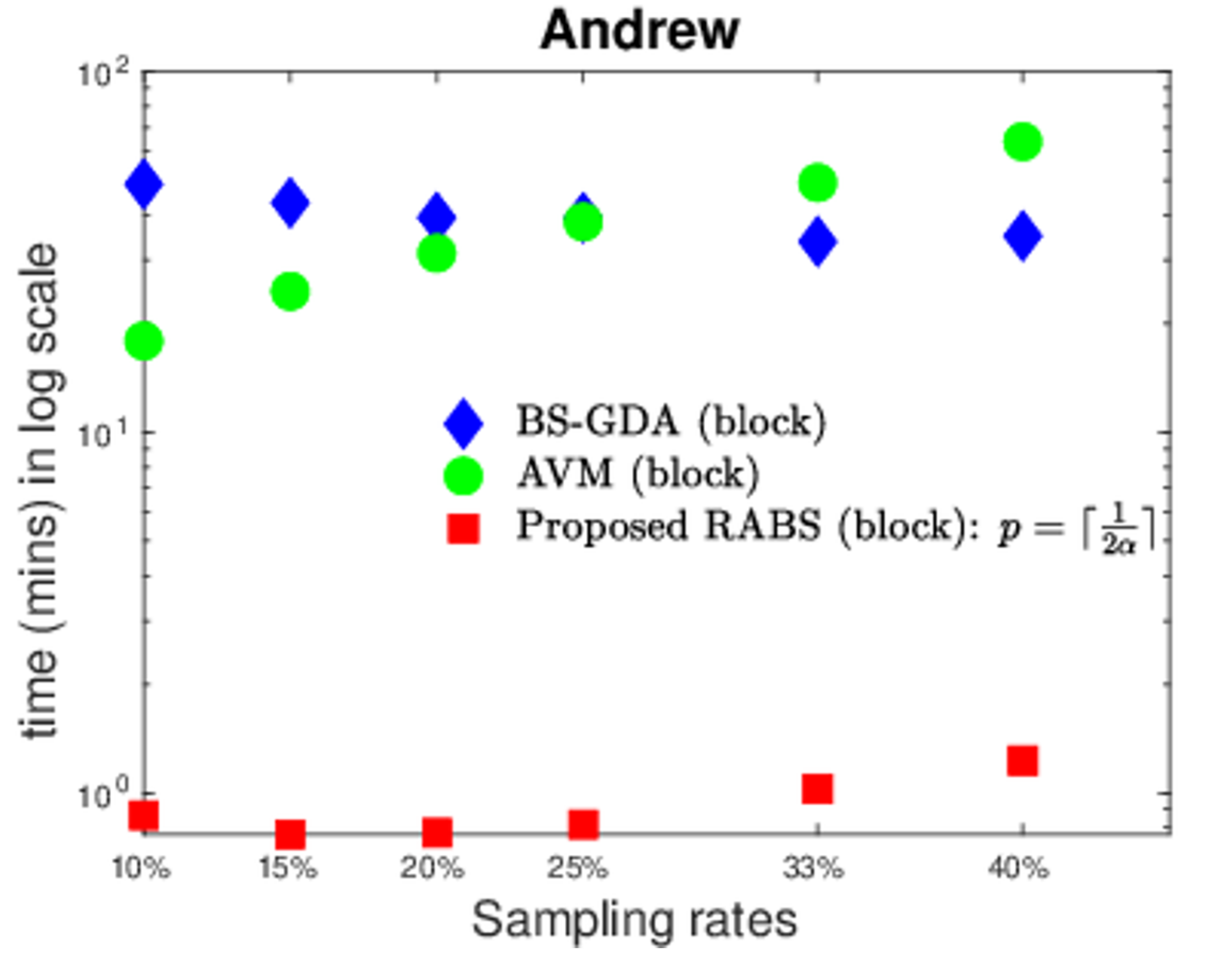}
     \end{subfigure}
\caption{Runtime comparison between graph-based greedy sampling algorithms - AVM \cite{jayawant2021_avm}, BS-GDA \cite{bai2020_bsgda} and \hyperref[algo:block_sampling]{RABS} ($p = \ceil{\frac{1}{2\alpha}}$) in a block-based framework (block size: $2^{6} \times 2^6 \times 2^6$)}
\label{fig:runtime_comparison} 
\end{figure*}

\subsection{Effect of selecting the parameter \texorpdfstring{$p$}{TEXT} on PSNR and runtime}
\label{subsec:p_effect_exps}
We compare the PSNR and runtime when PCs are sampled using \hyperref[algo:block_sampling]{RABS} for different values of $p$ in \autoref{fig:p_vs_alpha_psnr}.
We can observe the PSNR  values increase as we increase $p$ from $1$ to $5$ at a lower sampling rate (eg. 10\%-20\%). At higher sampling rates ($>30\%$), we see that increasing $p$ does not affect the PSNR, but the runtime still increases  (\autoref{fig:p_vs_alpha_psnr}). These results are consistent with our complexity analysis in \autoref{subsec:complexity_analysis} that complexity increases with  $p$. Our numerical results show that at higher sampling rates, a lower $p$ is sufficient to achieve a good PSNR in a computationally efficient way.

We also compare the proposed methods to select $p$ based on sampling rates \eqref{eqn:p_estimate_uni} and \eqref{eqn:p_estimate_rand}. \autoref{fig:fixed_p_vs_adap_p} shows that selecting $p$ using \eqref{eqn:p_estimate_uni} or \eqref{eqn:p_estimate_rand} performs equally well compared to fixing $p$ to an arbitrarily high value, but with a much lower runtime. Among the proposed methods, selecting $p$ based on a uniform sampling model \eqref{eqn:p_estimate_uni} results in a marginally better PSNR with a small overhead in runtime compared to selecting $p$ based on random sampling model \eqref{eqn:p_estimate_rand}, which is consistent with the observation $p_{\text{rand}} \leq p_{\text{uni}}$, from  \autoref{subsec:effect_of_p}. In  our sub-sequent experiments we will use \eqref{eqn:p_estimate_uni}, i.e.,  $p = \ceil{1/(2\alpha)}$.

\subsection{Global sampling vs Block-based sampling}
\label{subsec:self_vs_noself}
We compare the performance of proposed \hyperref[algo:fast_sampling]{RAGS} with  \hyperref[algo:block_sampling]{RABS} algorithms. As shown in \autoref{fig:global_vs_block_self}, RAGS results in high PSNR at all sampling rates, and RABS is faster by more than an order of magnitude while achieving similar PSNR compared to global sampling. The results in \autoref{fig:global_vs_block_self} are consistent with our complexity analysis in \autoref{subsec:block_sampl_complexity}. Furthermore, we show the effect of self-loops in RABS at a fixed sampling rate for different values of $p$ in \autoref{fig:self_vs_noself}. We can observe that the introduction of self-loops with optimally derived weights in \eqref{eqn:self_loop_weight} yields a noticeable improvement in PSNR (around $0.1$dB - $0.2$dB) at all values of $p$ compared to RABS without self-loops. 
Although there is a small overhead to compute self-loop weights, the gain in reconstruction accuracy justifies this additional computational cost. Another interesting observation is that RAGS results in poor PSNR compared to RABS for lower values of $p$, especially when the sampling rate is low (\autoref{fig:self_vs_noself}). 
This is not surprising because when we sample a large graph with highly localized interpolators (lower value of $p$) at lower sampling rates, the relative importance i.e.,  $\langle \qv_{i}^{(p)}, \qv_{j}^{(p)} \rangle = 0, i \neq j $ for most of the points. 

\begin{table*}[ht]
\begin{center}
\caption{Comparison of reconstruction accuracy of the proposed \hyperref[algo:block_sampling]{RABS} algorithm (block size = $2^{6} \times 2^{6} \times 2^{6}$, $p = \ceil{1/(2\alpha)}$) with uniform sampling, GFPR sampling, BS-GDA, and AVM. Block-based sampling is adopted for AVM and BS-GDA. }
\begin{adjustbox}{width=\textwidth,center}
\label{tab:psnr_comparison_table}
\begin{tabular}{|c|lll|lll|ccc|ccc|ccc|ccc|}
\hline
PC                   & \multicolumn{3}{c|}{\begin{tabular}[c]{@{}c@{}}Uniform sampling \cite{sridhara2022_chromapc}\\ (global)\\ -PSNR (dB)\end{tabular}} & \multicolumn{3}{c|}{\begin{tabular}[c]{@{}c@{}}GFPR sampling \cite{chen2017_pcsamplgraph}\\ (global)\\ -PSNR (dB)\end{tabular}} & \multicolumn{3}{c|}{\begin{tabular}[c]{@{}c@{}}BS-GDA \cite{bai2020_bsgda} (block)\\ - PSNR (dB)\end{tabular}} & \multicolumn{3}{c|}{\begin{tabular}[c]{@{}c@{}}AVM \cite{jayawant2021_avm} (block)\\ - PSNR (dB)\end{tabular}} & \multicolumn{3}{c|}{\begin{tabular}[c]{@{}c@{}}Proposed \hyperref[algo:block_sampling]{RABS} \\ (block)\\ -PSNR (dB)\end{tabular}} & \multicolumn{3}{c|}{\begin{tabular}[c]{@{}c@{}}Proposed \hyperref[algo:fast_sampling]{RAGS} \\ (global)\\ -PSNR (dB)\end{tabular}} \\
\multicolumn{1}{|l|}{} & 10\%                          & 25\%                         & 33\%                                  & 10\%                              & 25\%                             & 33\%                             & \multicolumn{1}{l}{10\%}     & \multicolumn{1}{l}{25\%}    & \multicolumn{1}{l|}{33\%}    & \multicolumn{1}{l}{10\%}    & \multicolumn{1}{l}{25\%}   & \multicolumn{1}{l|}{33\%}   & \multicolumn{1}{l}{10\%}     & \multicolumn{1}{l}{25\%}     & \multicolumn{1}{l|}{33\%}    & \multicolumn{1}{l}{10\%}     & \multicolumn{1}{l}{25\%}     & \multicolumn{1}{l|}{33\%}     \\  \hline
\textit{longdress}   & 26.16                         & 31.55                        & 34.03                                  & 27.29                             & 30.20                            & 31.30                            & 28.59                        & 32.38                       & 32.19                       & \textbf{28.70}              & 32.12                      & 33.54                      & 28.38                        & 32.46                        & 34.29                       & 28.40                        & \textbf{32.55}               & \textbf{34.31}               \\
\textit{loot}        & 33.73                         & 40.20                        & 42.88                                 & 34.75                             & 37.97                            & 39.13                            & 36.39                        & 40.84                       & 40.49                       & \textbf{36.44}              & 40.47                      & 42.03                      & 36.12                        & 40.92                        & 43.00                       & 36.18                        & \textbf{41.05}               & \textbf{43.04}               \\
\textit{redandblack} & 32.16                         & 38.13                        & 40.52                                 & 33.19                             & 36.54                            & 37.78                            & 34.82                        & 38.73                       & 38.63                       & \textbf{34.90}              & 38.52                      & 39.92                      & 34.63                        & 38.92                        & 40.61                       & 34.73                        & \textbf{39.00}               & \textbf{40.65}               \\
\textit{soldier}     & 29.40                         & 35.72                        & 38.54                                 & 30.60                             & 33.96                            & 35.27                            & 32.18                        & 36.58                       & 36.21                       & \textbf{32.24}              & 36.17                      & 37.76                      & 31.87                        & 36.64                        & 38.74                       & 32.02                        & \textbf{36.76}               & \textbf{38.78}               \\ \hline
\textit{Sarah}       & 39.71                         & 46.66                        & \textbf{48.47}                        & 41.82                             & 45.72                            & 47.18                            & 43.79                        & 47.09                       & 46.97                       & \textbf{43.82}              & 47.15                      & 48.28                      & 43.77                        & 47.12                        & 48.29                       & 43.72                        & \textbf{47.15}               & 48.31                        \\
\textit{David}       & 40.24                         & 48.12                        & 50.18                                 & 42.81                             & 47.00                            & 48.39                            & \textbf{45.16}               & 48.78                       & 48.74                       & 45.10                       & 48.72                      & 49.89                      & 44.97                        & 48.95                        & 50.19                       & 45.03                        & \textbf{49.00}               & \textbf{50.21}               \\
\textit{Andrew}      & 27.51                         & 34.08                        & \textbf{36.53}                        & 29.02                             & 32.64                            & 34.08                            & 30.76                        & 34.72                       & 34.51                       & 30.66                       & 34.48                      & 35.91                      & \textbf{30.74}               & 35.06                        & 36.49                       & 30.74                        & \textbf{35.07}               & 36.51                        \\
\textit{Phil}        & 30.04                         & 37.42                        & 39.42                                 & 32.47                             & 36.19                            & 37.52                            & \textbf{34.65}               & 38.09                       & 37.89                       & 34.55                       & 38.03                      & 39.20                      & 34.41                        &  38.25              & 39.51              & 34.46                        & \textbf{38.27 }                       & \textbf{39.54}                        \\
\textit{Ricardo}     & 38.76                         & 45.63                        & 47.39                                 & 40.77                             & 44.12                            & 45.54                            & \textbf{42.62}               & 45.86                       & 45.93                       & 42.34                       & 45.88                      & 47.03                      & 42.55                        & 46.16                        & 47.40                       & 42.27                        & \textbf{46.17}               & \textbf{47.42}    \\ \hline          
\end{tabular}
\end{adjustbox}
\end{center}
\end{table*}

\begin{table*}[ht]
\begin{center}
\caption{Runtime comparison (reported in minutes) of the \hyperref[algo:fast_sampling]{RAGS}($p = \ceil{1/(2\alpha)}$), \hyperref[algo:block_sampling]{RABS} (block size = $2^{6} \times 2^{6} \times 2^{6}$) with BS-GDA, and AVM in the same block-based framework.}
\label{tab:runtime_comparison_table}
\captionsetup{justification=centering}
\begin{tabular}{|c|c|ccl|ccl|ccl|ccc|}
\hline
PC                   & \begin{tabular}[c]{@{}c@{}}number\\ of points\end{tabular} & \multicolumn{3}{c|}{\begin{tabular}[c]{@{}c@{}}BS-GDA \cite{bai2020_bsgda} (block)\\   - runtime (min)\end{tabular}} & \multicolumn{3}{c|}{\begin{tabular}[c]{@{}c@{}}AVM \cite{jayawant2021_avm} (block)\\ - runtime (min)\end{tabular}} & \multicolumn{3}{c|}{\begin{tabular}[c]{@{}c@{}}Proposed \hyperref[algo:fast_sampling]{RAGS} \\ (global)\\ - runtime (min)\end{tabular}} & \multicolumn{3}{c|}{\begin{tabular}[c]{@{}c@{}}Proposed \hyperref[algo:block_sampling]{RABS} \\ (block)\\ - runtime (min)\end{tabular}} \\
                     &                                                            & 10\%                           & 25\%                           & 33\%                          & 10\%                         & 25\%                         & 33\%                         & 10\%                                & 25\%                                & 33\%                                & 10\%                                & 25\%                                & 33\%                                 \\ \hline
\textit{longdress}   & 806K                                                       & 22.19                          & 18.72                          & 14.93                         & 9.00                         & 17.30                        & 21.91                        & 10.41                               & 29.01                               & 40.46                       & \textbf{0.44}                       & \textbf{0.38}                       & \textbf{0.47}                               \\
\textit{loot}        & 792K                                                       & 19.90                          & 16.84                          & 13.98                         & 8.32                         & 15.80                        & 19.78                        & 9.90                                & 28.23                               & 39.68                   &  \textbf{0.40}                       & \textbf{0.35}                       & \textbf{0.43}                                   \\
\textit{redandblack} & 654K                                                       & 16.87                          & 14.13                          & 11.33                         & 6.96                         & 13.42                        & 16.96                        & 6.61                                & 18.68                               & 26.30                      & \textbf{0.32}                       & \textbf{0.28}                       & \textbf{0.36}                              \\
\textit{soldier}     & 1063K                                                      & 28.47                          & 23.90                          & 19.18                         & 11.58                        & 22.40                        & 28.44                        &  18.55                               & 51.88                               & 72.96                        &  \textbf{0.62}                       & \textbf{0.54}                       & \textbf{0.67}                              \\\hline
\textit{Sarah}       & 1088K                                                      & 35.98                          & 28.85                          & 24.96                         & 13.60                        & 28.07                        & 36.63                        & 19.49                               & 54.57                               & 76.52                        & \textbf{0.67}                       & \textbf{0.60}                       & \textbf{0.79}                               \\ 
\textit{David}       & 1241K                                                      & 40.32                          & 32.39                          & 27.97                         & 15.25                        & 31.50                        & 40.21                        & 24.58                               & 69.78                               & 97.94                      &  \textbf{0.78}                       & \textbf{0.71}                       & \textbf{0.87}                                 \\
\textit{Andrew}      & 1299K                                                      & 48.74                          & 39.32                          & 33.78                         & 17.90                        & 38.25                        & 49.19                        & 28.34                               & 79.87                               & 112.03                       & \textbf{0.86}                       & \textbf{0.81}                       & \textbf{1.02}                               \\
\textit{Phil}        & 1344K                                                      & 46.53                          & 37.02                          & 33.58                         & 17.51                        & 36.60                        & 46.95                        & 30.58                               & 85.28                               & 119.83                        & \textbf{0.91}                       & \textbf{0.82}                       & \textbf{1.03}                              \\
\textit{Ricardo}     & 919K                                                       & 32.10                          & 25.96                          & 22.47                         & 12.08                        & 25.23                        & 32.47                        & 13.58                               & 38.42                               & 54.05                      & \textbf{0.54}                       & \textbf{0.50}                       & \textbf{0.64}    \\ \hline                          
\end{tabular}
\end{center}
\end{table*}

\begin{table}[ht]
\begin{center}
\caption{Runtime comparison with fast, non-greedy sampling algorithms - uniform \cite{sridhara2022_chromapc} and geometry feature preserving random (GFPR) PC sampling \cite{chen2017_pcsamplgraph}. }
\begin{adjustbox}{width=0.5\textwidth,center}
\label{tab:runtime_uni_rand}
\begin{tabular}{|c|ccl|ccc|ccc|}
\hline
dataset              & \multicolumn{3}{c|}{\begin{tabular}[c]{@{}c@{}}Uniform sampling \cite{sridhara2022_chromapc}\\ (global)\\ -avg. runtime (s)\end{tabular}} & \multicolumn{3}{c|}{\begin{tabular}[c]{@{}c@{}}GFPR sampling \cite{chen2017_pcsamplgraph}\\ (global)\\ -avg. runtime (s)\end{tabular}} & \multicolumn{3}{c|}{\begin{tabular}[c]{@{}c@{}}Proposed \hyperref[algo:block_sampling]{RABS} \\ (block)\\ -avg. runtime (s)\end{tabular}} \\
\multicolumn{1}{|l|}{} & 10\%                               & 25\%                               & 33\%                              & 10\%                                & 25\%                                & 33\%                               & \multicolumn{1}{l}{10\%}       & \multicolumn{1}{l}{25\%}       & \multicolumn{1}{l|}{33\%}       \\ \hline
\textit{8i}          & \textbf{0.008}                     & \textbf{0.009}                     & \textbf{0.009}                    & 0.10                                & 0.14                                & 0.16                               & 26.97                          & 23.48                          & 29.28                          \\
\textit{MVUB}        & \textbf{0.01}                      & \textbf{0.01}                      & \textbf{0.01}                     & 0.18                                & 0.24                                & 0.28                               & 45.52                          & 41.44                          & 52.51                         \\ \hline
\end{tabular}
\end{adjustbox}
\end{center}
\end{table}

\subsection{Comparison with uniform and graph sampling algorithms}

\label{subsec:comparison}
We conducted a comprehensive comparison of the proposed \hyperref[algo:fast_sampling]{RAGS} and \hyperref[algo:block_sampling]{RABS} algorithms with other existing techniques uniform sampling \cite{sridhara2022_chromapc} (our preliminary work), graph-based geometry feature preserving random (GFPR) sampling for PCs\cite{chen2017_pcsamplgraph}, and existing ED-free graph signal sampling algorithms. 
We consider two deterministic, greedy sampling algorithms - AVM \cite{jayawant2021_avm} and BS-GDA \cite{bai2020_bsgda} for comparison. 
We compare the reconstruction accuracy of the above sampling algorithms using the same reconstruction method in \eqref{eqn:recon_minimization} for all sampling methods
We will evaluate uniform sampling only at $10\%, 25\%, \text{and } 33\%$ because uniform sampling does not allow other sampling rates, which is one of its major drawbacks. Uniform and GFPR sampling methods scale well to PC graphs, and they can be implemented without partitioning the PC into blocks. While AVM and BS-GDA provide good reconstruction accuracy and are considered fast among existing graph signal sampling algorithms,  they still cannot scale to the graph sizes considered in this work. 
For that reason, we apply them in a block-based manner, using the same block sizes as in our RABS.  
While evaluating RABS, we fix the block size to $2^{6} \times 2^{6} \times 2^{6}$ and we include self-loops with weights computed from \eqref{eqn:self_loop_weight}. We use  $p = \ceil*{\frac{1}{2 \alpha}}$ as in \eqref{eqn:p_estimate_uni}. 

 \autoref{tab:psnr_comparison_table} and \autoref{fig:psnr_comparison} compare these sampling algorithms on PC attribute reconstruction for different sampling rates. Evidently, GFPR sampling, geared towards preserving geometry features, leads to a subpar reconstruction of attributes. While uniform sampling results in competitive reconstruction at higher sampling rates, reconstruction accuracy at lower rates is poor. Not surprisingly, the sampling set obtained from graph signal sampling algorithms - AVM and BS-GDA, resulted in a fairly faithful reconstruction of attributes at lower rates. However, the performance of BS-GDA degrades at higher sampling rates. 
RAGS marginally outperforms RABS at all sampling rates. Overall, the proposed sampling algorithms consistently outperformed BS-GDA and AVM, specifically at higher sampling rates. The drop in PSNR values of the RABS compared to AVM at lower sampling rates could be because of underestimating the value of $p$ from \eqref{eqn:p_estimate_uni}, when $\alpha$ is small. 

\autoref{fig:runtime_comparison} compares the runtime of RAGS and RABS with AVM and BS-GDA. RABS is up to 50 times faster on both datasets compared to AVM and BS-GDA in a similar block-based sampling framework. We show the detailed runtime comparison of graph signal sampling algorithms in \autoref{tab:runtime_comparison_table}. Notably, even RAGS, which is a global algorithm, has a runtime similar to that of a block-based implementation of AVM and BS-GDA, especially at lower rates.
For completeness, we include the runtime comparison between RABS and existing fast PC sampling approaches - uniform and GFPR sampling in \autoref{tab:runtime_uni_rand}. Unsurprisingly, uniform and GFPR sampling algorithms are much faster than RABS but result in poor reconstruction accuracy, as shown in  \autoref{tab:psnr_comparison_table}. 

In summary, the proposed \hyperref[algo:fast_sampling]{RAGS} and  \hyperref[algo:block_sampling]{RABS} algorithms  have comparable reconstruction accuracy with the fastest graph-based sampling algorithms. While RAGS has comparable runtime, RABS is 20-50 times faster than the existing graph signal sampling algorithms. The efficient and adaptive computation of interpolators based on the sampling rate allows us to scale the proposed algorithm to larger point clouds and higher sampling rates, while the block-based approach allows us to exploit sparsity and locality efficiently. Also, the proposed RABS outperforms uniform and GFPR sampling methods by a large margin in reconstruction accuracy.
\subsection{PC attribute compression via sampling}
\label{subsec_exp_chroma}
We provide an end-to-end evaluation of PC attribute sampling in a compression system. Similar to image and video encoding, we sample the chrominance channel before compression to exploit the fact that errors in chrominance are perceptually less significant than errors in luminance and reconstruct the sampled signal at the decoder end. In our previous work, we proposed a chroma sub-sampling framework for PC attribute compression through uniform sampling \cite{sridhara2022_chromapc}. As discussed, the uniform sampling approach will result in sub-optimal signal recovery for a given sampling rate. In this section, we evaluate chroma subsampling using the proposed \hyperref[algo:block_sampling]{RABS} algorithm. In all our experiments, we sample the chrominance signal at a sampling rate of 25\%.  We evaluate chroma subsampling on the 8i PC dataset. We use a transform coding system comprised of the Region adaptive hierarchical transform (RAHT) \cite{Ricardo2016_raht}, and run-length Golomb-Rice entropy coding \cite{malvar2006adaptive}. We compare the rate-distortion (RD) curves, and report bitrate savings and average psnr gain using Bjontegaard metric \cite{bjontegaard2001}. We employ uniform sampling and \hyperref[algo:block_sampling]{RABS} to sample the chroma channels and compare the coding gains. The point-wise distortion in RGB space is calculated using,
\begin{equation}
    \resizebox{.85\hsize}{!}{$\text{PSNR}_{rgb}=-10\log_{10}\left(\frac{\Vert \rv -\hat{\rv}\Vert_{2}^{2}+\Vert \gv-\hat{\gv} \Vert_{2}^{2}+\Vert \bv-\hat{\bv}\Vert_{2}^{2}}{(3N255^2)}\right)$},
\end{equation}
where $N$ is the number of points in the full-resolution PC.

\begin{table}[htb]
\centering
\captionsetup{justification=centering}
\caption{Compression results when chroma channels are sampled at $25\%$ using \hyperref[algo:block_sampling]{RABS} and uniform sampling}
\label{tab:chroma_table}
\begin{tabular}{|c|cc|cc|}

\hline
sequence    & \multicolumn{2}{c|}{Proposed RABS}                                                                                                               & \multicolumn{2}{c|}{Uniform sampling}                                                                                                                \\ \hline
            & \multicolumn{1}{c|}{\begin{tabular}[c]{@{}c@{}}avg. PSNR \\ gain (dB)\end{tabular}} & \begin{tabular}[c]{@{}c@{}}bitrate \\ saving (\%)\end{tabular} & \multicolumn{1}{c|}{\begin{tabular}[c]{@{}c@{}}avg. PSNR \\ gain (dB)\end{tabular}} & \begin{tabular}[c]{@{}c@{}}bitrate \\ saving (\%)\end{tabular} \\ \hline
longdress   & \multicolumn{1}{c|}{\textbf{0.64}}                                                  & \textbf{10.55}                                                 & \multicolumn{1}{c|}{0.39}                                                           & 4.04                                                           \\ \hline
redandblack & \multicolumn{1}{c|}{\textbf{0.48}}                                                  & \textbf{6.49}                                                  & \multicolumn{1}{c|}{0.27}                                                           & 5.22                                                           \\ \hline
loot        & \multicolumn{1}{c|}{0.11}                                                           & 2.81                                                           & \multicolumn{1}{c|}{\textbf{0.12}}                                                  & \textbf{2.98}                                                  \\ \hline
soldier     & \multicolumn{1}{c|}{\textbf{0.11}}                                                  & \textbf{2.19}                                                  & \multicolumn{1}{c|}{0.09}                                                           & 1.99                                                           \\ \hline
\end{tabular}
\end{table}

The RD curves for two PC sequences in the RGB domain are shown in \autoref{fig:chroma_exps} and show a considerable gain in attribute coding. 
As shown in the \autoref{tab:chroma_table}, the bitrate savings (bpp) for \textit{redandblack} and \textit{longdress} sequences are around $7\%-11\%$, and for \textit{loot} and \textit{soldier} sequences it around $2\%-3\%$ (not shown due to space). The coding gains of the proposed sampling over uniform sampling are larger for \textit{longdress} and \textit{redandblack} sequence than for  \textit{loot} and \textit{soldier}, since the latter sequences are smoother.

\begin{figure}[t]
     \centering
        \begin{subfigure}[b]{0.24\textwidth}
            \centering
            \includegraphics[width=\linewidth]{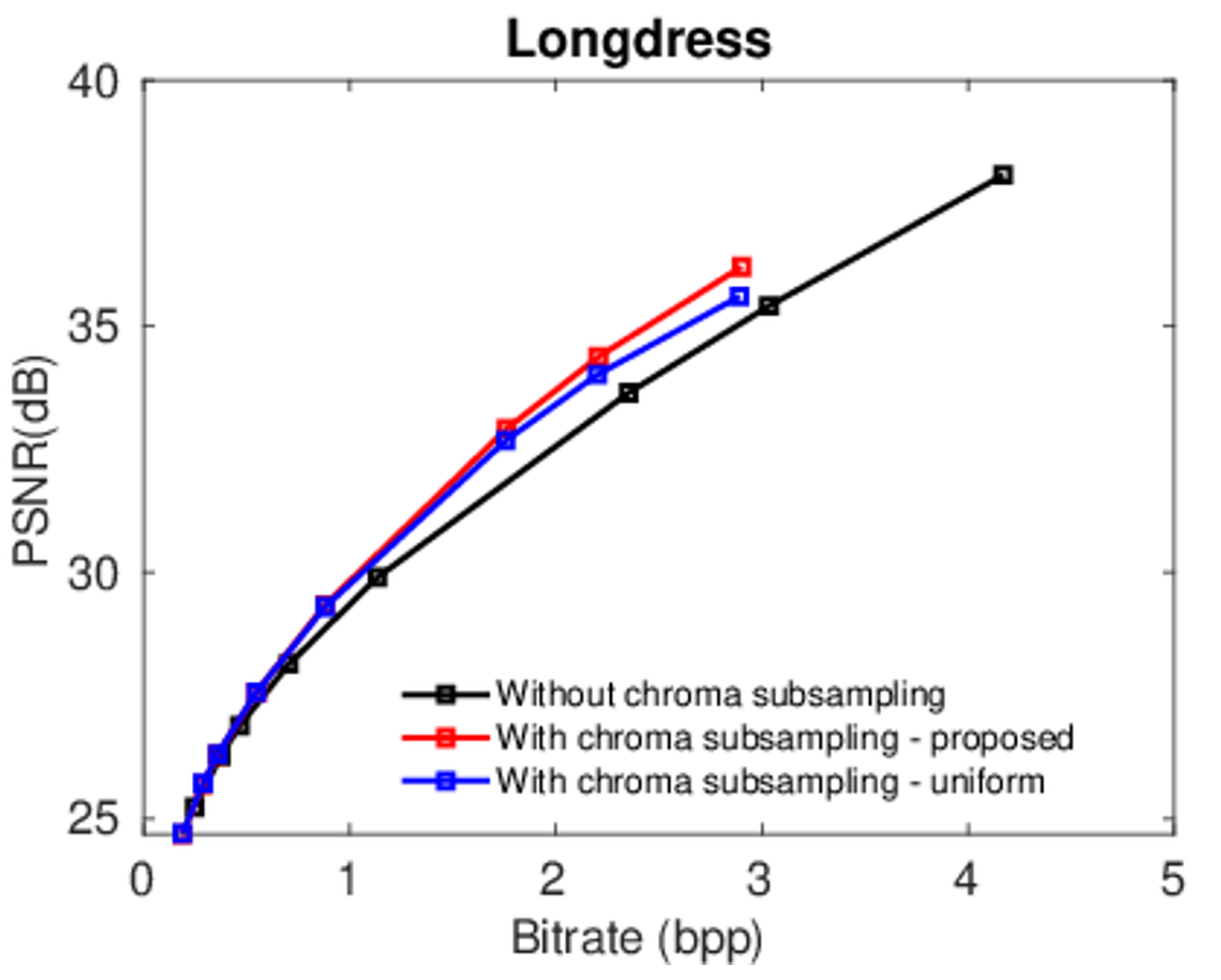}
        \end{subfigure}
      \begin{subfigure}[b]{0.24\textwidth}
         \centering
         \includegraphics[width=\linewidth]{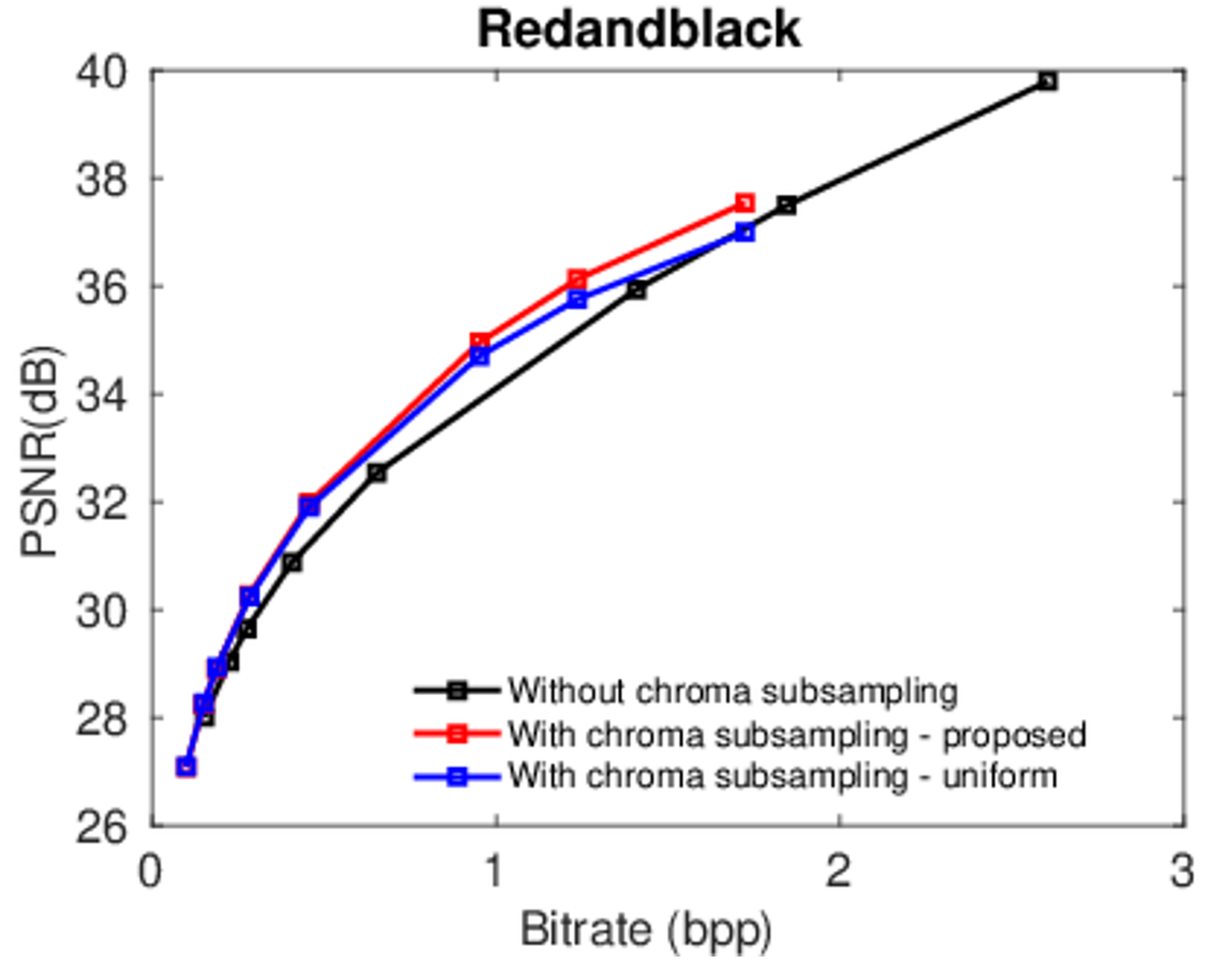}
     \end{subfigure}

\caption{Rate-distortion curves for PC color attribute compression when chrominance signal is sampled at a rate of 25\%.}
\label{fig:chroma_exps}
\end{figure}

\section{Conclusion}
\label{sec:conclusion}
\noindent We have proposed scalable graph-based sampling algorithms for PC attributes.
We first proposed a reconstruction-aware global sampling set selection algorithm based on a fast reconstruction method.  We showed how to construct interpolators that can be used for sampling. The proposed interpolators can be efficiently implemented without sacrificing performance by exploiting the interplay between the localization (sparsity) of the interpolators and the sampling rate. We proposed a method to select the suitable localization parameter for sampling.
To further reduce the sampling complexity, we proposed a block-based algorithm, where we sample blocks independently. We proposed a method to locally compute interpolators by adding self-loops to the sub-graphs, leading to an order of magnitude reduction in complexity and runtime,  while maintaining the good performance of global sampling. 
We evaluated the proposed global and block sampling algorithms on large PCs. We showed that our proposed sampling algorithms are up to $50$ times faster than existing graph-based sampling algorithms by ensuring similar or better reconstruction. Our sampling algorithm outperforms commonly used PC sampling techniques - uniform and geometry feature preserving sampling with respect to reconstruction accuracy. Additionally, we showcased a practical application in PC attribute compression through chroma subsampling, which resulted in bitrate savings of up to $11\%$.
\bibliographystyle{IEEEbib}
\bibliography{strings}
\appendix

\subsection{Complexity of RAGS}
\label{app:complexity}
Since $\Qm(1) = \Zm$, we can compute $\Qm(p)$ using the recursive formula $\Qm(t) = \Zm + \Qm(t-1)\Zm $, and thus the complexity of computing $\Qm(p)$ depends on the cost of sparse matrix product and sparse matrix addition. For $p=2$, note that the matrix $\Zm^2$ has at most $\Oc(\Bar{d}^2N)$ entries, and each of its entries is computed by the inner product of two $\Bar{d}$-sparse vectors, thus the total complexity of computing $\Zm^2$ is $\Oc(\Bar{d}^3 N)$. The matrix addition has complexity proportional to the number of entries in $\Zm^2$, which is $\Oc(\Bar{d}^2N)$. In general for $t \geq 2$, the product $\Qm(t-1)\Zm$ has $\Oc(\Bar{d}^t N)$ non-zero entries and each entry is the inner product of a $\Bar{d}$-sparse vector with a $\Bar{d}^{t-1}$-sparse vector, which has cost $\Oc(\Bar{d}^{t-1})$ per inner product, thus the computation of the product $\Qm(t-1)\Zm$ has cost $\Oc(\Bar{d}^{2t}N)$. The complexity of the matrix addition is dominated by the number of entries of the densest matrix, which is $\Oc(\Bar{d}^t N)$. Because the recursion must be applied $p-1$ times, the total complexity of computing $\Qm(p)$ is dominated by the cost matrix multiplication applied $p$ times, which is given by $\Oc(p\Bar{d}^{2p} N)$.
To compute the inner products between the columns of $\Qm(p)$, we can compute the product $(\Qm(p))^{\top}\Qm(p)$. This matrix has $\Bar{d}^{2p}N$ non-zero entries, and each of them is computed by an inner product between $\Bar{d}^p$-sparse vectors, resulting in total complexity of $\Oc(\Bar{d}^{3p}N)$.
 The complexity of the preparation step is thus dominated by the computation of $(\Qm(p))^{\top}\Qm(p)$ over the computation of  $\Qm(p)$, which results in a total complexity of $\Oc(\Bar{d}^{3p} N)$.
 \begin{remark}\label{remark_app_complexity_densegraph}
   Note that if $\Bar{d}$ is large (e.g., for irregular graphs) or if $p$ is large,  $\Oc(\Bar{d}^{2t}N)$ may significantly overestimate the number of entries in  $\Qm(t)$, in which case we say that the number of entries is $\Oc(N^2)$, resulting in a worst-case complexity $\Oc(\Bar{d}^p N^2)$. In this work, we can assume that  $\Bar{d}$ is relatively small with respect to $N$ since the point clouds are dense and points belong to the 3D integer grid, and $p$ can be kept small by using techniques developed in \autoref{sec:sampling_algo_development}.
 \end{remark}
Regarding sampling, since the inner products between columns of $\Qm(p)$ are already computed, the greedy sampling set update only involves adding  $\Oc(N)$ inner products and solving the maximization problem in \eqref{eqn:greedy_update}. This has  $\mathcal{O}(N)$ cost per iteration. Since we have to repeat this step $s$ times, the overall complexity of the sampling set update is $\mathcal{O}(sN)$.

\subsection{Proof of \autoref{lemma_equivalent_interpolators} }
\label{app:reconstruction_proof}
%
We divide the proof in two Lemmas, one for each equality.
\begin{lemma}
   The closed form solution \eqref{eqn:recon_soln} for reconstruction can be expressed equivalently in terms of   $\Zm$ defined in \eqref{eqn:1_hop_operator}
    \begin{equation}
        \hat{\fv} = (\Id - \Pm_{\Rc} \Zm)^{-1} \Pm_{\Sc} \fv=  \begin{bmatrix}
                        \fv_{\Sc}\\
                        - \Lm_{\Rc\Rc}^{-1} \Lm_{\Rc\Sc} \fv_{\Sc}
                    \end{bmatrix}.
    \end{equation}
\end{lemma}
\begin{proof}
Using  $\Pm_{\Sc}\fv=\begin{bmatrix}
                        \fv_{\Sc}^{\top} &
                        \zerov^{\top}
                    \end{bmatrix}^{\top}$, and the block matrix inverse formula \cite{horn2012_matrix} we can show that
\begin{equation}
    (\Id - \Pm_{\Rc} \Zm)^{-1} \Pm_{\Sc} \fv= \begin{bmatrix}
                        \Id_{\Sc \Sc} \\
                        (\Id_{\Rc \Rc} - \Zm_{\Rc \Rc})^{-1} \Zm_{\Rc \Sc} 
                    \end{bmatrix} \fv_{\Sc}.
\end{equation}
It is easy to check that  $(\Id_{\Rc \Rc} - \Zm_{\Rc \Rc})^{-1} \Zm_{\Rc \Sc}=- \Lm_{\Rc\Rc}^{-1} \Lm_{\Rc\Sc}$, which combined with \eqref{eqn:recon_soln} produces the desired result.
\end{proof}






\begin{lemma}
\label{lemma:convergence}
    We have the  identity
    \begin{equation}
        (\Id - \Pm_{\Rc} \Zm)^{-1} \Pm_{\Sc} = \sum_{l=0}^{\infty} (\Pm_{\Rc} \Zm)^{l} \Pm_{\Sc}.
    \end{equation}
\end{lemma}

\begin{proof}
 The truncated interpolator matrix is given by
\begin{equation}
    \sum_{l=0}^{k} (\Pm_{\Rc} \Zm)^{l} \Pm_{\Sc}.
\end{equation}
This sum will converge  if and only if  $\Vert \Pm_{\Rc}\Zm \Vert<1$, which is proven in \autoref{lemma:non_expansive}. Therefore, from Neumann series for matrices (Th. 5.6.12 of \cite{horn2012_matrix}) we have that $\lim_{i \rightarrow \infty} \sum_{l=0}^{i} (\Pm_{\Rc} \Zm)^{l} = (\Id - \Pm_{\Rc} \Zm)^{-1}$.
\end{proof}
\begin{lemma}
\label{lemma:non_expansive}
    The operator  $\Pm_{\Rc}\Zm$ is strictly non-expansive, that is, $\Vert \Pm_{\Rc}\Zm \xv\Vert <1$ for all $\xv$.
\end{lemma}
\begin{proof}
Note that $\Vert \Pm_{\Rc}\Zm\xv\Vert^2 = \sum_{i \in \Rc}(\langle \bv_i,\xv\rangle)^2 \leq \sum_{i =1}^N(\langle \bv_i,\xv\rangle)^2 = \Vert \Zm \xv\Vert^2 \leq \Vert \xv \Vert^2$, where $\bv_i$ is defined in Appendix \ref{app_matrix_inequality}.
The matrices $\Zm$ and $\Pm_{\Rc}$ are both non expansive. In fact,  for connected graphs  $\Vert \Zm \xv \Vert = \Vert \xv \Vert$ if and  only if  $\xv = \alpha \cdot \mathbf{1}$ \cite{horn2012_matrix}. 
If $\xv = \alpha \mathbf{1}$ for some $\alpha \neq 0$, the first inequality becomes strict because $\Zm$ is a non negative irreducible matrix. If $\xv \neq \alpha \mathbf{1}$
the second inequality becomes strict.
%
%
%
\subsection{Proof of \autoref{lemma:error_p}}
We consider the error vector between the $p^{th}$ order approximation of the reconstructed signal and the original signal $\ev^{(p)} = \hat{\fv}^{(p)} - \fv$.
Using the triangular inequality, we arrive at
\begin{equation}
\label{eqn:error_bound}
    \norm{\ev^{(p)}} \leq \norm{\hat{\fv}^{(p)} - \hat{\fv}} + \norm{\hat{\fv} - \fv}.
\end{equation}
We will focus only on the first term \eqref{eqn:error_bound}, which depends on $p$.  From  the definition of $\hat{\fv}^{(p)}$ we can show that
\begin{equation}
   \hat{\fv}^{(p)}  - \hat{\fv}  = \Pm_{\Rc}\Zm   (\hat{\fv}^{(p-1)} - \hat{\fv})  = (\Pm_{\Rc}\Zm)^{p}   (\hat{\fv}^{(0)} - \hat{\fv}),
\end{equation}
which implies the error bound
\begin{equation}
\label{eqn:intermediate_norm1}
    \norm{ \hat{\fv}^{(p)}  - \hat{\fv}} \leq \norm{\Pm_{\Rc}\Zm}^{p} \norm{\hat{\fv}^{(0)} - \hat{\fv}}.
\end{equation}
We can bound the error of the initialization
by 
\begin{equation}
    \norm{\hat{\fv}^{(0)} - \hat{\fv}} \leq \Vert \hat{\fv}^{(0)} -\fv\Vert  + \Vert \fv- \hat{\fv} \Vert \leq \Vert \fv \Vert +\Vert \fv- \hat{\fv}\Vert, 
\end{equation}
where the second inequality is due to $\Vert \hat{\fv}^{(0)} -\fv\Vert = \Vert \fv_{\Rc}\Vert \leq \Vert \fv\Vert$. Replacing back into \eqref{eqn:error_bound} we have
\begin{align}
 \norm{\ev^{(p)}} &\leq   \norm{\Pm_{\Rc}\Zm}^{p}\Vert \fv \Vert +  (1+\norm{\Pm_{\Rc}\Zm}^{p})\Vert \fv- \hat{\fv} \Vert.
\end{align}
We conclude by dividing both sides by $\norm{\fv}$,  taking expectation with respect to $\Sc$, and applying Cauchy–Schwarz.

\subsection{Proof of \autoref{thrm:selectin_p}}
\label{app_matrix_inequality}
Following \cite{tropp2015introduction} we use $\Vert \Pm_{\Rc} \Zm \Vert^2 = \Vert \Zm^{\top} \Pm_{\Rc} \Zm \Vert$ and 
\begin{equation}
    \Zm^{\top} \Pm_{\Rc} \Zm = \sum_{i=1}^N \delta_i \bv_i \bv_i^{\top},
\end{equation}
where $\lbrace \delta_i \rbrace_{i=1}^N$ are independent identically distributed  Bernoulli random variables with $\Prob(\delta_i=1) = 1-\alpha$, and $\bv_i^{\top}$ is the $i$th row of $\Zm$. 
Let $q = p/2$, and since $p \geq 2$, Jensen's inequality implies the lower bound
\begin{align}
    ({1-\alpha})^q = {\Vert \EE \Zm^{\top}\Pm_{\Rc}\Zm \Vert}^q\leq \EE \Vert\Zm^{\top}\Pm_{\Rc}\Zm  \Vert^q.
\end{align}
Replacing back $q = p/2$ we get the desired lower bound
\begin{equation}
    (\sqrt{1-\alpha})^p \leq  \EE \Vert\Pm_{\Rc}\Zm  \Vert^p.
\end{equation}
For the upper bound, we  apply the triangular inequality for the operator norm and Minkowski's inequality for the expectation
\begin{align}
    &\left(\EE \Vert\Zm^{\top}\Pm_{\Rc}\Zm  \Vert^q\right)^{1/q} \leq \Vert \EE \Zm^{\top}\Pm_{\Rc}\Zm \Vert +\\
    &\left(\EE\Vert \Zm^{\top}\Pm_{\Rc}\Zm - \EE\Zm^{\top}\Pm_{\Rc}\Zm  \Vert^q\right)^{1/q}\\
    &= (1-\alpha) + \left(\EE\Vert \Zm^{\top}\Pm_{\Rc}\Zm - (1-\alpha)\Zm^{\top}\Zm  \Vert^q\right)^{1/q} \\
    &\leq \left(\sqrt{1-\alpha} + \left(\EE\Vert \Zm^{\top}\Pm_{\Rc}\Zm - (1-\alpha)\Zm^{\top}\Zm  \Vert^q\right)^{\frac{1}{2q}}\right)^2.
\end{align}

\end{proof}

\end{document}